\documentclass[superscriptaddress,amsmath,amssymb,aps,pra,letterpaper,tightenlines,notitlepages,12pt]{revtex4-2}

\pdfoutput=1

\usepackage{amsmath}
\usepackage{amsfonts}
\usepackage{amssymb}
\usepackage{mathtools}
\usepackage{graphicx}
\usepackage{hyperref}
\usepackage{xfrac}
\usepackage{physics}
\usepackage{xcolor}
\usepackage{babel}[english]
\usepackage{multirow}

\providecommand{\U}[1]{\protect\rule{.1in}{.1in}}

\hypersetup{colorlinks=true,citecolor=blue,linkcolor=blue,filecolor=blue,urlcolor=blue,breaklinks=true}
\newtheorem{theorem}{Theorem}

\newtheorem{definition}[theorem]{Definition}

\newtheorem{lemma}[theorem]{Lemma}

\newtheorem{proposition}[theorem]{Proposition}
\newtheorem{remark}[theorem]{Remark}

\newenvironment{proof}[1][Proof]{\noindent\textbf{#1.} }{\ \rule{0.5em}{0.5em}}

\allowdisplaybreaks

\begin{document}

\title{Improved sample complexity bound for\\sample-based Lindbladian simulation}

\author{Siheon Park}
\thanks{These authors contributed equally to this work.}
\affiliation{Department of Physics and Astronomy, Seoul National University, Seoul 08826, Republic of Korea}

\author{Youngjin Seo}
\thanks{These authors contributed equally to this work.}
\affiliation{School of Computational Sciences, Korea Institute for Advanced Study, Seoul 02455, Republic of Korea}

\author{Byeongseon Go}
\affiliation{Department of Physics and Astronomy, Seoul National University, Seoul 08826, Republic of Korea}

\author{Dhrumil~Patel}
\affiliation{Department of Computer Science, Virginia Tech, Alexandria, Virginia 22305, USA}
\affiliation{Department of Computer Science, Cornell University, Ithaca, New York 14850, USA}

\author{Mark~M.~Wilde}
\email{wilde@cornell.edu}
\affiliation{School of Electrical and Computer Engineering, Cornell University, Ithaca, New York 14850, USA}

\author{Hyukjoon Kwon}
\email{hjkwon@kias.re.kr}
\affiliation{School of Computational Sciences, Korea Institute for Advanced Study, Seoul 02455, Republic of Korea}

\begin{abstract}
We establish improved sample-complexity bounds for sample-based Lindbladian simulation based on the Wave Matrix Lindbladization (WML) algorithm. For a jump operator $L$ with dimension $d$, we derive an explicit non-asymptotic sample complexity bound $n_d^*(t,\varepsilon) \le \left( \frac{2d+3}{8} \right) \left\|L\right\|_\infty^2 \left( \frac{t^2}{\varepsilon} \right)$, holding for simulation time $t$ and error $\varepsilon$. This refines the dimension dependence of the best previously known bound, $O(d^2 t^2/\varepsilon)$, from [Go \textit{et al.}, Quantum~Sci.~Tech.~\textbf{10}, 045058 (2025)]. Remarkably, we show that this dimensional overhead can be entirely avoided when $\| L\|_\infty^2 = O(1/d)$, a condition satisfied with high probability for random Lindblad operators, yielding a typical-case sample complexity of $O(t^2/\varepsilon)$. On the other hand, in the worst case, we show that WML necessarily requires $\Omega(dt^2/\varepsilon)$ samples by constructing an explicit example with a rank-one Lindblad operator. Our results reveal a sharp dichotomy between typical and adversarial sample complexities in Lindbladian simulation, thereby strengthening the theoretical foundations of sample-based quantum algorithms.
\end{abstract}

\maketitle

\tableofcontents

\section{Introduction}

The accurate simulation of quantum systems is a cornerstone of quantum information science \cite{breuer2002TheoryOpenQuantum, nielsen2010QuantumComputationQuantum, wiseman2009QuantumMeasurementControl, weiss2012QuantumDissipativeSystems, feynman1982SimulatingPhysicsComputers}, with broad implications ranging from quantum chemistry \cite{wecker2014GatecountEstimatesPerforming, babbush2015ChemicalBasisTrotterSuzuki, reiher2017ElucidatingReactionMechanisms, poulin2014TrotterStepSize} to condensed-matter physics~\cite{raeisi2012QuantumcircuitDesignEfficient, daley2014QuantumTrajectoriesOpen}. Hamiltonian simulation~\cite{low2017OptimalHamiltonianSimulation, berry2015SimulatingHamiltonianDynamics, low2019HamiltonianSimulationQubitization, childs2018FirstQuantumSimulation} has been used to model the unitary dynamics of closed quantum systems, enabling the study of phenomena such as quantum phase transitions and entanglement growth. However, many realistic quantum systems are inherently open because they interact with external environments, leading to dissipation and decoherence. The dynamics of such systems can be described by the Gorini--Kossakowski--Sudarshan--Lindblad (GKSL) master equation, commonly referred to as Lindbladian dynamics, which provides a general framework for Markovian non-unitary time evolution arising from environmental coupling~\cite{breuer2002TheoryOpenQuantum, gorini1976CompletelyPositiveDynamical, lindblad1976GeneratorsQuantumDynamical}.

Despite substantial progress in quantum simulation, the efficient simulation of Lindbladian dynamics remains comparatively underexplored relative to Hamiltonian simulation. While this is partly because quantum computation has conventionally been formulated in terms of unitary operations, the simulation of open quantum dynamics has increasingly been studied in the contexts of dissipative quantum computing~\cite{poyatos1996QuantumReservoirEngineering, cubitt2015StabilityLocalQuantuma, zhan2026RapidQuantumGround}, quantum error mitigation~\cite{vandenberg2023ProbabilisticErrorCancellation}, and Gibbs sampling~\cite{jin2024QuantumSimulationPartial, lin2025DissipativePreparationManybody}, to name a few. Several approaches have been proposed to simulate open quantum systems, including higher-order Trotterization and product-formula methods~\cite{li2023SimulatingMarkovianOpen, burdine2024EfficientSimulationOpen,childs2023EfficientSimulationSparse}, \emph{qDrift}-like randomization~\cite{kato2024ExponentiallyAccurateOpen, Peng2025,Chen2025randomizedmethod}, quantum trajectory methods known as unraveling~\cite{dalibard1992WavefunctionApproachDissipative, dum1992MonteCarloSimulation, carmichael1993OpenSystemsApproach, cao2025DynamicallyOptimalUnraveling}, embeddings of Lindbladian dynamics into effective Hamiltonian simulations~\cite{cleve2019EfficientQuantumAlgorithms, li2023SimulatingMarkovianOpen, ding2024SimulatingOpenQuantum, low2025OptimalQuantumSimulation}, and vectorized linear-combination techniques for superoperators~\cite{kamakari2022DigitalQuantumSimulation, schlimgen2022QuantumSimulationLindblad}. These methods exploit unitary primitives to emulate dissipative evolution, thereby narrowing the gap between closed- and open-system simulation and enabling the study of quantum dynamics in realistic settings.

Wave Matrix Lindbladization (WML) was recently proposed as a sample-based Lindbladian simulation algorithm \cite{patel2023WaveMatrixLindbladization, patel2023WaveMatrixLindbladizationa}, and it offers potential advantages analogous to those of sample-based Hamiltonian simulation~\cite{lloyd2014QuantumPrincipalComponent,kimmel2017HamiltonianSimulationOptimal, wei2024SimulatingNoncompletelyPositive, go2025SamplebasedHamiltonianLindbladian}. In WML, one has access to multiple copies of a program state that encodes the target Lindbladian dynamics. The goal is to implement the Lindbladian evolution $e^{ \mathcal{L} t}$ approximately for a desired evolution time $t \geq 0$, where the Lindbladian operator $\mathcal{L}$ is the generator of the Lindbladian dynamics. While Hamiltonian terms and multiple dissipative terms could be included in more general Lindbladian dynamics, in this work, we focus instead on a single Lindblad operator $L$, which serves as the foundation for understanding the algorithms for general Lindbladians presented in Ref.~\cite{patel2023WaveMatrixLindbladizationa}.

The sample complexity of sample-based quantum simulation is defined as the number of copies of a program state, $\pi$, required to realize the target quantum channel $e^{ \mathcal{L} t}$ up to a desired error tolerance $\varepsilon \in [0, 1]$. The state-of-the-art sample complexity upper bound of WML was shown as $O(d^2 t^2/\varepsilon)$ for $L$ acting on a $d$-dimensional Hilbert space and normalized as $\left\| L\right\|_2 = 1$~\cite{go2025SamplebasedHamiltonianLindbladian}. Although $d$ remains small for most physically relevant systems governed by local Lindbladians, a comprehensive analysis for arbitrary, not necessarily local, Lindbladians is crucial for a more complete theoretical understanding of simulating open quantum dynamics~\cite{kliesch2011DissipativeQuantumChurchTuring, cubitt2015StabilityLocalQuantuma, cleve2019EfficientQuantumAlgorithms, xu2019ExtremeDecoherenceQuantum, lin2025DissipativePreparationManybody}. As the algorithm-independent sample complexity lower bound $\Omega(t^2/\varepsilon)$ has no explicit dependence on $d$~\cite{go2025SamplebasedHamiltonianLindbladian}, the resulting $d^2$-gap between the WML upper bound and the general lower bound leaves open the possibility of sample-based quantum algorithms, or even hybrid quantum--classical algorithms, that outperform WML. This possibility is particularly relevant in light of the fact that quantum state tomography requires $\Omega(d^2/\varepsilon^2)$ samples to achieve trace distance $\varepsilon$~\cite{haah2016SampleoptimalTomographyQuantum, haah2017SampleoptimalTomographyQuantum, yuen2023ImprovedSampleComplexity, patel2023WaveMatrixLindbladizationa, anshu2023SurveyComplexityLearning}.

\bigskip

\noindent
\textbf{Our contribution.} In this work, we establish a tighter upper bound on the sample complexity of WML for Lindbladian simulation, namely
\begin{equation}
    O(d \left\|L\right\|_\infty^2 t^2/\varepsilon),
\end{equation} 
showing that its dimension dependence is at most linear, given that $\left\|L\right\|_\infty^2 \leq \left\| L\right\|_2^2 = 1$. This result narrows the gap with the previously known sample-complexity lower bound~\cite{go2025SamplebasedHamiltonianLindbladian} while preserving the $t^2/\varepsilon$ scaling. Crucially, reducing the dimension dependence from quadratic to linear restores a provably strict separation between the sample complexity of Lindbladian simulation and that of program-state tomography. 
Furthermore, using random matrix theory, we show that the dimensional overhead can even disappear for typical random instances. More precisely, Lindblad operators randomly drawn from the Frobenius-normalized Ginibre ensemble satisfy $\left\|L\right\|_\infty^2=O(1/d)$ with high probability, in which case the WML sample-complexity upper bound becomes dimension-independent in the large-$d$ limit. In contrast, we construct an explicit worst-case example for which the sample complexity of WML necessarily scales linearly with $d$, indicating an irreducible dimensional gap between typical and worst-case sample complexities. The resulting upper and lower bounds are summarized in Table~\ref{table:results}.

\begin{table}[htbp]
\centering
\renewcommand{\arraystretch}{1.4} 
\begin{tabular}{c||c|c}
    & Upper bound & Lower bound \\
    \hline\hline
    \multirow{2}{*}{General case}
    & \rule{0pt}{3.5ex} $n_d^\ast(t,\varepsilon) =  O\!\left(\frac{d^2t^2}{\varepsilon}\right)$ \cite[Theorem 9]{go2025SamplebasedHamiltonianLindbladian} \rule[-2ex]{0pt}{0pt}
    & \multirow{2}{*}{$n_d^\ast(t,\varepsilon)= \Omega\!\left(\frac{t^2}{\varepsilon}\right)$ \cite[Theorem 12]{go2025SamplebasedHamiltonianLindbladian}} \\
    \cline{2-2}
    & \rule{0pt}{3.5ex} $n_d^\ast(t,\varepsilon)= O\!\left(\frac{d\| L \|_\infty^2t^2}{\varepsilon}\right)$ \textbf{[Theorem~\ref{cor:samp-comp-upper-bnd-WML}]} \rule[-2ex]{0pt}{0pt}
    & \\
    \hline
    Typical case
    & \rule{0pt}{3.5ex} $\overline{n}_d(t,\varepsilon,\delta)= O\!\left(\frac{t^2}{\varepsilon}\right)$ \textbf{[Theorem~\ref{thm:average-case-WML}]}
    & \\
    \hline
    Worst case
    & \rule{0pt}{3.5ex} $n_d^\ast(t,\varepsilon)= O\!\left(\frac{dt^2}{\varepsilon}\right)$ \textbf{[Theorem~\ref{cor:samp-comp-upper-bnd-WML}]}
    & $n_{\mathrm{WML}}(t,\varepsilon;\mathcal{L})=\Omega\!\left(\frac{dt^2}{\varepsilon}\right)$ \textbf{[Theorem~\ref{prop:rankone_lower}]} \\
    \hline
    \begin{tabular}{@{}c@{}}Special case \end{tabular}
    & $n_d^\ast(t,\varepsilon)= O\!\left(\frac{t^2}{\varepsilon}\right)$
    & $n_d^\ast(t,\varepsilon)= \Omega\!\left(\frac{t^2}{\varepsilon}\right)$ \cite[Theorem 12]{go2025SamplebasedHamiltonianLindbladian} \\
\end{tabular}
\caption{Summary of sample-complexity upper and lower bounds of sample-based simulation. The row labeled ``Special case'' corresponds to local Lindbladians where $d = O(1)$ is the dimension of the system on which the Lindblad operator acts.}
\label{table:results}
\end{table}

\bigskip

\noindent
\textbf{Structure of the paper.}
In Section~\ref{sec:background}, we introduce basic notation and review sample-based Lindbladian simulation, especially WML. We also recall the formal definition of the sample complexity of sample-based Hamiltonian and Lindbladian simulation~\cite[Definition~1]{go2025SamplebasedHamiltonianLindbladian}.  In Section~\ref{sec:upper-bound-WML}, we establish an improved upper bound on the sample complexity of Lindbladian simulation algorithms, based on a rigorous analysis of WML. In Section~\ref{sec:average-case}, we analyze the typical-case sample complexity, where the jump operators are drawn from an ensemble on its parameter space, i.e., the bipartite Hilbert space. We then prove concentration bounds for the Frobenius-normalized operator norm and derive the resulting regime-dependent scaling of the WML sample complexity. In Section~\ref{sec:worst-case}, we construct an explicit rank-one worst-case instance of $\left\|L\right\|_\infty^2=\Theta(1)$, which yields the WML sample-complexity that necessarily scales linearly in $d$. In Section~\ref{sec:implication}, we describe implications of dimensional factor reduction for quantum information theory, as well as its application to practical scenarios. Finally, in Section~\ref{sec:conclusion}, we provide some concluding remarks and directions for future research.

\section{Background}

\label{sec:background}

\subsection{Notation}
In this section, we provide some notations and simple propositions that will be used throughout the paper. Consider a quantum system $S$ associated with a Hilbert space $\mathcal{H}_{S}$. We write the set of quantum states (density operators) on this space as $\mathcal{D}(\mathcal{H}_{S})$. When the dimension $d$ needs to be specified, we use $\mathcal{H}^{d}_S$ and $\mathcal{D}(\mathcal{H}^{d}_S)$, or $\mathcal{H}_d$ and $\mathcal{D}(\mathcal{H}_d)$ if the label $S$ can be omitted.

Let the trace and adjoint of an operator $X$ be $\operatorname{Tr}[X]$ and $X^{\dag}$, respectively. For $p \in [1,\infty)$, the Schatten $p$-norm of $X$ is given by
\begin{align}\label{eq:schatten-norm}
    \|X\|_{p} \coloneqq \left( \operatorname{Tr}\!\left[\left(XX^{\dag}\right)^{\frac{p}{2}} \right]\right)^{\frac{1}{p}}.
\end{align}
Our focus will primarily be on the trace norm ($\left\| \cdot \right\|_{1}$) and the operator norm ($\left\| \cdot \right\|_{\infty}$).

To quantify the distinguishability of two states $\rho, \sigma \in \mathcal{D}(\mathcal{H}_{S})$, we evaluate the normalized trace distance:
\begin{align}
    D_{\rm tr} (\rho, \sigma) \coloneqq \frac{1}{2}\| \rho - \sigma \|_{1}.
\end{align}
This metric captures the maximum difference in measurement probabilities for the two states~\cite[Eq.~(9.22)]{wilde2017QuantumInformationTheory}. The prefactor $1/2$ bounds the distance to the interval $[0,1]$, although we may occasionally omit it for brevity. To measure state similarity, we use the fidelity~\cite{uhlmann1976TransitionProbabilityState},
\begin{align}\label{eq: fidelity definition}
    F(\rho, \sigma) \coloneqq \left\| \sqrt{\rho} \sqrt{\sigma} \right\|_1^2 = \operatorname{Tr}\!\left[\sqrt{\sqrt{\rho}\sigma\sqrt{\rho}}\right]^2.
\end{align}

When comparing two quantum channels $\mathcal{N}$ and $\mathcal{M}$, i.e., completely positive (CP) and trace-preserving (TP) maps, we use the normalized diamond distance~\cite{kitaev1997QuantumComputationsAlgorithms}:
\begin{align}
    \frac{1}{2}\| \mathcal{N} - \mathcal{M} \|_{\diamond} \coloneqq \frac{1}{2}\sup_{\rho \in \mathcal{D}(\mathcal{H}_{R}\otimes \mathcal{H}_{S})}\left\|(\mathcal{I}_{R}\otimes\mathcal{N})(\rho) - (\mathcal{I}_{R}\otimes\mathcal{M})(\rho) \right\|_{1},
\end{align}
where $R$ is a reference system of arbitrary dimension and $\mathcal{I}_{R}$ is the identity map on $R$. The optimization over all bipartite input states bounds the value to the interval $[0,1]$. By standard arguments, it suffices to restrict the dimension of the reference system such that $\dim(\mathcal{H}_{R}) = \dim(\mathcal{H}_{S})$~\cite[Theorem~9.1.1]{wilde2017QuantumInformationTheory}.

Finally, let $\operatorname{Tr}_{S_2}[\rho]$ denote the partial trace of a bipartite state $\rho \in \mathcal{D}(\mathcal{H}_{S_1} \otimes \mathcal{H}_{S_2})$ over system $S_2$, and let $\mathbb{I}_{S} \coloneqq \sum_{i}|i\rangle\!\langle i|_{S}$ be the identity operator acting on $S$. We define the SWAP operator between $S_1$ and $S_2$ as
\begin{align}\label{eq:swap-def}
    \operatorname{SWAP}_{S_1S_2} \coloneqq \sum_{i,j}|i\rangle\!\langle j|_{S_1} \otimes|j\rangle\!\langle i|_{S_2}.
\end{align}

\subsection{Sample-based quantum simulation}

\label{sec:sample-complex-def}

Sample-based simulation~\cite{lloyd2014QuantumPrincipalComponent, kimmel2017HamiltonianSimulationOptimal, wei2024SimulatingNoncompletelyPositive,go2025SamplebasedHamiltonianLindbladian} aims to implement target dynamics by means of a fixed quantum channel ${\cal P}^{(n)}$ acting on an arbitrary input state $\rho$ and $n$ identical copies of a program state $\pi$, up to an error $\varepsilon$. It may therefore be viewed as a form of $\varepsilon$-programmable quantum simulation~\cite{nielsen1997ProgrammableQuantumGate, banchi2020ConvexOptimizationProgrammable, jing2025ProgrammableOpenQuantum}, with the key structural constraint being that the program is provided as i.i.d.\ copies of a single quantum state. The main quantity of interest is then the \textit{minimum} number of program-state copies, $n$, required to simulate the dynamics for evolution time~$t$ with error at most $\varepsilon$ in the diamond distance. This motivates the definition of sample complexity as the minimum $n$ for which the target dynamics can be approximated to an accuracy of $\varepsilon$, where the minimization is taken over all $(n+1)$-partite completely positive and trace-preserving (CPTP) maps~\cite{go2025SamplebasedHamiltonianLindbladian}.

\begin{definition}[Sample complexity of sample-based quantum simulation]
\label{def: sample complexity}
The sample complexity of sample-based quantum simulation is defined as the minimum number $n^{\ast}_d(t,\varepsilon)$ of program states $\pi$ required to realize a quantum channel $e^{{\cal L}_\pi t}$, a one-parameter semi-group generated by a linear superoperator $\mathcal{L}_\pi$ corresponding to $\pi$ and acting on $\mathcal{B}(\mathcal{H}_d)$, within an error $\varepsilon$ in the normalized diamond distance. Formally, the sample complexity $n^{\ast}_d(t,\varepsilon)$ is defined as
\begin{align}
    n^{\ast}_d(t,\varepsilon)
    & \coloneqq  \inf_{\mathcal{P}^{(n)}\in\rm{CPTP}
    }\left\{  n\in\mathbb{N}:\sup_{\mathcal{L}_\pi}\frac{1}{2}\left\Vert \mathcal{P}^{(n)}
    \circ\mathcal{A}_{\pi^{\otimes n}}-e^{\mathcal{L}_\pi t}\right\Vert
    _{\diamond}\leq\varepsilon\right\},
\end{align}
where the appending channel $\mathcal{A}_{\pi^{\otimes n}}$ for the quantum state $\kappa$ on an arbitrary input state~$\zeta$ is defined as
\begin{equation}
\mathcal{A}_{\kappa^{\otimes n}}(\zeta)\coloneqq \zeta\otimes\kappa^{\otimes
n}.
\end{equation}
\end{definition}

Sample-based simulation has several features that make it attractive for near-term noisy intermediate-scale (NISQ) devices. Unlike gate-based simulation, it replaces long coherent gate sequences with the preparation of i.i.d.~copies of a program state, a resource that may be easier to generate and stabilize experimentally. Moreover, because the program is supplied as i.i.d.~copies, the sample-based setting is compatible with purification or distillation-type preprocessing of the program resource, which can be used to boost its fidelity prior to the simulation~\cite{bennett1996MixedStateEntanglement, bravyi2005UniversalQuantumComputation, yuan2023VirtualQuantumResource}.

A further advantage is a form of quantum privacy~\cite{aaronson2009QuantumCopyProtectionQuantum, kimmel2017HamiltonianSimulationOptimal, patel2023WaveMatrixLindbladization}. As in quantum learning settings where useful processing can be performed without fully revealing the underlying quantum data, sample-based simulation allows one to implement dynamics without exposing a complete description of the program state. Indeed, reconstructing an unknown program state $\pi$ by tomography generally requires substantially more copies than are needed simply to use $\pi$ as a simulation resource. We note that the program state for simulating a general Lindbladian needs to be a $d^2$-dimensional quantum state, as explained in detail in Appendix~\ref{appendix:state2dyn}. Meanwhile, quantum state tomography requires at least ${\Omega}(d^{2}r/\varepsilon^2)$ copies for a $d^2$-dimensional rank-$r$ quantum state~\cite{haah2016SampleoptimalTomographyQuantum, haah2017SampleoptimalTomographyQuantum, yuen2023ImprovedSampleComplexity, patel2023WaveMatrixLindbladizationa,  anshu2023SurveyComplexityLearning}, whereas the simulation task can require fewer copies. Consequently, sample-based algorithms enable accurate simulation without fully learning the program state that encodes the dynamics.

A lower bound on sample complexity follows from Def.~\ref{def: sample complexity}, by comparing quantum state and channel discrimination. Here we summarize the result on the lower bound established in Ref.~\cite[Theorem 12]{go2025SamplebasedHamiltonianLindbladian}, namely
\begin{equation}\label{eq:lower_bound_result}
n^{\ast}_d(t,\varepsilon) \geq 10^{-4}\,\frac{t^{2}}{\varepsilon},
\end{equation}
for the case of dephasing Lindbladian on projected two-dimensional Hilbert space, given the condition $\varepsilon \in [0, \min\{0.039,\,0.013 t\}]$. We say a sample-based quantum simulation is optimal if the sample complexity of such an algorithm has a matching upper bound, i.e., $O(t^2/\varepsilon)$. For example, Hamiltonian dynamics can be generated in the sample-based simulation framework by coherent Lindbladians parameterized as $\mathcal{L}_\pi = -i[\pi,\cdot]$. The overall normalization factor can be absorbed into the evolution time $t$ without changing the scale of the dynamics. The sample complexity lower bound in this case is also $\Omega(t^2/\varepsilon)$~\cite{wei2024SimulatingNoncompletelyPositive, go2025SamplebasedHamiltonianLindbladian}, with a matching upper bound achieved by the Lloyd--Mohseni--Rebentrost (LMR) protocol, also known as density matrix exponentiation~\cite{lloyd2014QuantumPrincipalComponent, kimmel2017HamiltonianSimulationOptimal, wei2024SimulatingNoncompletelyPositive, go2025SamplebasedHamiltonianLindbladian}. Therefore, we can conclude that the sample complexity of sample-based Hamiltonian simulation is $n^{\ast}_d(t,\varepsilon) \in \Theta({t^{2}}/{\varepsilon})$. Recently, the WML algorithm~\cite{patel2023WaveMatrixLindbladization, patel2023WaveMatrixLindbladizationa} for Lindbladian simulation was shown to achieve a matching sample complexity upper bound up to a constant factor~\cite{go2025SamplebasedHamiltonianLindbladian}..

\subsection{Review of wave matrix Lindbladization}
In the WML algorithm, we set $\mathcal{P}^{(n)}$ to be a sequence of $n$ repeated interactions between the input state and a fresh copy of the program state, analogous to the construction used in sample-based Hamiltonian simulation~\cite{kimmel2017HamiltonianSimulationOptimal,wei2024SimulatingNoncompletelyPositive,go2025SamplebasedHamiltonianLindbladian}. For the sake of clarity, we will consider a simple Lindbladian of the following form,
\begin{equation}
    \label{eq:lindblad_dynamics}
    \mathcal{L}^{(1)}(\rho) \coloneqq L \rho L^\dagger - \frac{1}{2} \left\{L^\dagger L, \rho \right\}. 
\end{equation}
This can be extended to general Lindbladian dynamics with multiple dissipative and coherent terms~\cite{patel2023WaveMatrixLindbladizationa}. Considering a time step $\Delta\coloneqq \frac{t}{n}$, then a single step of WML is given by
\begin{equation}
\widetilde{e^{\mathcal{L}\Delta}}
(\rho)\coloneqq\operatorname{Tr}_{23}[e^{\mathcal{M}\Delta}  (\rho\otimes\pi_{L})],\label{eq:WML-actual-one-step}
\end{equation}
where $\pi = \ket{L} \!\bra{L}$ is the bipartite and pure program state in a $d\times
d$-dimensional Hilbert space. Here, we adopt the following convention for vectorization of operators:
\begin{align}
|{L}\rangle & \coloneqq \left(  L\otimes I\right)  |\Gamma\rangle,\\
|\Gamma\rangle & \coloneqq \sum_{i=1}^{d}|i\rangle|i\rangle,
\end{align}
and we assume that the program states are normalized, i.e., $\left\| L\right\|_2=1$. Also, ${\mathcal{M}}$ is a fixed Lindbladian generator independent of the target Lindbladian dynamics, given as
\begin{align}
\mathcal{M}(\omega)  & \coloneqq M\omega M^{\dag}-\frac{1}{2}\left\{  M^{\dag} M,\omega\right\}  , \label{eq:WML-M-op} \\
M  & \coloneqq \frac{1}{\sqrt{d}} \left(  I_{1}\otimes|\Gamma\rangle\!\langle\Gamma|_{23}\right)  \left(
\mathrm{SWAP}_{12}\otimes I_{3}\right). \label{eq:WML-M-jump}
\end{align}
Throughout the paper, we will refer to $\mathcal{M}$ as the WML Lindbladian operator, and we note that an algorithm for implementing it has been proposed in \cite{sims2025DigitalQuantumSimulations}. Using the general identity $A\otimes B\ket{C} = \ket{ACB^T}$, we show that 
\begin{align}
M(\rho\otimes\pi_L)M^\dagger & =L\rho L^\dagger\otimes\ket{\Gamma}\!\bra{\Gamma}/d, \label{eq:1st-order-1}\\
\operatorname{Tr}_{23}M^\dagger M(\rho\otimes\pi_L) & =L^\dagger L\rho.
\label{eq:1st-order-2}
\end{align}
This leads to the fact that~\eqref{eq:WML-actual-one-step} is equivalent to $e^{\mathcal{L} \Delta}$ up to the first order in $\Delta$. Since $n$ compositions of $\widetilde{e^{\mathcal{L}\Delta}}$, i.e., $ \widetilde{e^{\mathcal{L}t}}\coloneq\left(\widetilde{e^{\mathcal{L}\Delta}}\right)^n$ is a CPTP map, the error bound
\begin{equation}
    \frac{1}{2}\left\|\widetilde{e^{\mathcal{L}t}}-e^{\mathcal{L}t}\right\|_\diamond\leq\varepsilon
\end{equation}
upper bounds the sample complexity $n_{d}^{\ast}(t,\varepsilon)$.
Noting that $\|\mathcal{M}\|_\diamond\leq2d$ and $\|\mathcal{L}\|_\diamond\leq2$~\cite[Appendix D]{go2025SamplebasedHamiltonianLindbladian}, the following upper bound on the sample complexity was proven in Ref.~\cite[Theorem 9]{go2025SamplebasedHamiltonianLindbladian}:
\begin{equation}
    n_{d}^{\ast}(t,\varepsilon)\leq 3d^2\frac{t^2}{\varepsilon}. \label{eq: old-upper}
\end{equation}

Although the scaling in $(t,\varepsilon)$ matches the fundamental lower bound in~\eqref{eq:lower_bound_result}, the $O(d^2)$ gap is not negligible for a large Hilbert space. In fact, the quadratic scaling in dimension compromises the aforementioned quantum privacy advantage of sample-based simulation methods~\cite{aaronson2009QuantumCopyProtectionQuantum, kimmel2017HamiltonianSimulationOptimal}, and it exhibits exponential scaling in the number of qubits for many-body systems if the Lindbladian is not local.

\section{Improved upper bound on the sample complexity of wave matrix Lindbladization} \label{sec:upper-bound-WML}
In this section, we derive a tighter upper bound on the non-asymptotic sample complexity of sample-based Lindbladian simulation along with an explicit coefficient. The first part of the main result is summarized as follows:
\begin{lemma} \label{Thm: error bound WML}
Suppose that the Lindbladian $\mathcal{L}$ acts nontrivially on a $d$-dimensional quantum system, where $d\in \mathbb{N}$ and $d\geq 2$. Let $t\geq0$, let $n\in\mathbb{N}$ be such that $n>2dt$, and let $\Delta\coloneqq\frac{t}{n}$. 
For every Lindbladian $\mathcal{L}$ satisfying $\left\|L\right\|_2=1$ in~\eqref{eq:lindblad_dynamics}, the error of wave matrix Lindbladization satisfies the following bound:
    \begin{equation} \label{eq:WML-err-bound}
        \frac{1}{2} \left\Vert e^{\mathcal{L}t}-\left(  \widetilde{e^{\mathcal{L}\Delta}}\right)
        ^{n}\right\Vert _{\diamond}\leq\frac{(2d+3)t^2}{8n}\left\|L\right\|_\infty^2,
    \end{equation}
where $\widetilde{e^{\mathcal{L}\Delta}}$ is a single-step evolution of WML in~\eqref{eq:WML-actual-one-step}. 
\end{lemma}

\begin{proof}
By taking $\Delta = \frac{t}{n}$, consider the following chain of inequalities:
\begin{align}
\frac{1}{2} \left\Vert e^{\mathcal{L}t}-\left(  \widetilde{e^{\mathcal{L}\Delta}}\right)
^{n}\right\Vert _{\diamond} &  = \frac{1}{2} \left\Vert \left(  e^{\mathcal{L}\Delta
}\right)  ^{n}-\left(  \widetilde{e^{\mathcal{L}\Delta}}\right)
^{n}\right\Vert _{\diamond}\\
&  \leq n \cdot \frac{1}{2}  \left\Vert e^{\mathcal{L}\Delta}-\widetilde{e^{\mathcal{L}\Delta}
}\right\Vert _{\diamond}\\
&  \leq n\frac{\Delta^2(2d+3)}{8} \left\|L\right\|_\infty^2  \\
&  =\frac{t^2(2d+3)}{8n} \left\|L\right\|_\infty^2,
\end{align}
where we have inductively applied the subadditivity of diamond distance (see, e.g., \cite[Appendix~B]{go2025SamplebasedHamiltonianLindbladian}) to obtain the first inequality. The second inequality follows from
Lemma~\ref{lem:single-step-WML-bound}, which then implies
\eqref{eq:WML-err-bound} after substituting $\Delta=\frac{t}{n}$.
\end{proof}

\medskip

From the error bound of the WML algorithm, we immediately conclude the following sample complexity upper bound:

\begin{theorem}[Upper bound on the sample complexity of WML]
\label{cor:samp-comp-upper-bnd-WML}
Lemma~\ref{Thm: error bound WML} implies  the following upper bound on the sample complexity of sample-based Lindbladian simulation:
\begin{equation}\label{eq:WML-sample-bound}
    n^{\ast}_d(t,\varepsilon) \leq  \frac{2d+3}{8} \left\|L\right\|_\infty^2 \left( \frac{t^2}{\varepsilon} \right).
\end{equation}
\end{theorem}

\begin{lemma}
\label{lem:single-step-WML-bound}
Suppose that $\mathcal{L}$ acts on a $d$-dimensional quantum system, where $d\in\mathbb{N}$ and $d\geq 2$. Let $t\geq0$, and let $n\in\mathbb{N}$ be such that
$n>2dt$, and set $\Delta\coloneqq\frac{t}{n}$, so that $\Delta<\frac{1}{2d}$.
For the channels $e^{\mathcal{L}\Delta}$ and $\widetilde{e^{\mathcal{L}\Delta
}}$ with time interval $\Delta$ in
\eqref{eq:WML-actual-one-step}, the following inequality holds:
    \begin{equation}
        \frac{1}{2}\left\Vert e^{\mathcal{L}\Delta}-\widetilde{e^{\mathcal{L}\Delta}}\right\Vert
        _{\diamond}\leq\frac{\Delta^2(2d+3)}{8} \left\|L\right\|_\infty^2 .
    \end{equation}
\end{lemma}

\begin{proof}
Let us label the $d$-dimensional input system on which the Lindbladian acts nontrivially as 1, and let us label the reference system as 0. Let us label the $d\times d$-dimensional program system as $2\otimes 3$, which is prepared in the program state $\pi_L$. We take a series expansion with respect to $\Delta$, which leads to
\begin{align}
    \widetilde{e^{\mathcal{L}\Delta}}(\rho)  & =\operatorname{Tr}_{23}[\left(
    \mathcal{I}\otimes e^{\mathcal{M}\Delta}\right)  (\rho\otimes\pi
    _{L})]\\
    & =\rho+\left(  \mathcal{I}\otimes\mathcal{L}\right)  (\rho)\Delta
    +\sum_{k=2}^{\infty}\operatorname{Tr}_{23}[\left(  \mathcal{I}
    \otimes\mathcal{M}^{k}\right)  (\rho\otimes\pi_{L})]\frac{\Delta^{k}}{k!},\label{eq:tmp-92}\\
    e^{\mathcal{L}\Delta}(\rho)  & =\rho+\left(  \mathcal{I}\otimes
    \mathcal{L}\right)  (\rho)\Delta+\sum_{k=2}^{\infty}\left(  \mathcal{I}
    \otimes\mathcal{L}\right)  ^{k}(\rho)\frac{\Delta^{k}}{k!}.
\end{align}
The first-order term of~\eqref{eq:tmp-92} is equal to $\left(\mathcal{I}\otimes\mathcal{L}\right)(\rho)$, which was proved in Ref.~\cite{patel2023WaveMatrixLindbladization} and reviewed in \eqref{eq:1st-order-1}--\eqref{eq:1st-order-2}. The key idea for this proof is to verify that $\operatorname{Tr}_{23}[M(\rho\otimes\pi_L)M^\dagger]$ and $\operatorname{Tr}_{23}[M^\dagger M(\rho\otimes\pi_L)]$ are equivalent to $L\rho L^\dagger$ and $L^\dagger L \rho$, as mentioned previously.
Exploiting the triangle inequality for the trace norm, we find that for every $\rho_{01}\in\mathcal{D}(\mathcal{H}^{d}\otimes\mathcal{H}^d)$,
\begin{align}
    \left\Vert e^{\mathcal{L}\Delta}(\rho_{01})-\widetilde{e^{\mathcal{L}\Delta}}(\rho_{01})\right\Vert_1
    &=\left\Vert \widetilde{e^{\mathcal{L}\Delta}}(\rho_{01})- e^{\mathcal{L}\Delta}(\rho_{01})\right\Vert_1\\
    & = \left\Vert \widetilde{e^{\mathcal{L}\Delta}}(\rho_{01})-\rho_{01}-\Delta\mathcal{L}(\rho_{01})+\sum_{k=2}^\infty\mathcal{L}^k(\rho_{01})\frac{\Delta^k}{k!}\right\Vert_1 \label{tmp:b}\\
    & \leq    \left\Vert\widetilde{e^{\mathcal{L}\Delta}}(\rho_{01})-\rho_{01}-\Delta\mathcal{L}(\rho_{01})\right\Vert_1 + \left\Vert\sum_{k=2}^\infty \frac{\Delta^k}{k!}\mathcal{L}^k\right\Vert_\diamond \label{tmp:c} \\
    & \leq    \frac{\Delta^2d}{2} \left\|L\right\|_\infty^2 + \left\Vert\sum_{k=2}^\infty \frac{\Delta^k}{k!}\mathcal{L}^k\right\Vert_\diamond \label{tmp:c3} \\
    & \leq   \frac{\Delta^2d}{2} \left\|L\right\|_\infty^2 + \sum_{k=2}^\infty\left\Vert \mathcal{L}\right\Vert^k_\diamond\frac{\Delta^k}{k!} \label{tmp:c2}\\
    & \leq \frac{\Delta^2d}{2}\left\|L\right\|_\infty^2 + (e^{2\Delta\left\|L\right\|_\infty^2}-2\Delta\left\|L\right\|_\infty^2-1) \label{tmp:d} \\
    & \leq \frac{\Delta^2d}{2} \left\|L\right\|_\infty^2 + \frac{3}{4}\Delta^2\left\|L\right\|_\infty^2 \\
    & = \frac{\Delta^2(2d+3)}{4} \left\|L\right\|_\infty^2.
    \label{tmp:f}
\end{align}
Here, \eqref{tmp:b} follows from the definition of~\eqref{eq:WML-actual-one-step}, and \eqref{tmp:c} follows from previous calculation that $\operatorname{Tr}_{23}\mathcal{M}(\rho_{01}\otimes\pi_L)=(\mathcal{I}\otimes\mathcal{L})(\rho_{01})$. \eqref{tmp:c3} follows from Lemma~\ref{lem:main-upper-lemma} and the fact that $0\leq1+\frac{e^{-x}-1}{x}\leq \frac{x}{2}$ for $x\in[0, \infty)$. \eqref{tmp:c2} follows from subadditivity of the diamond norm~\cite[Appendix B]{go2025SamplebasedHamiltonianLindbladian}. \eqref{tmp:d} follows from $\|\mathcal{L}\|_\diamond\leq2\left\|L\right\|_\infty^2$~\cite[Appendix D]{go2025SamplebasedHamiltonianLindbladian}. Finally,~\eqref{tmp:f} is derived using the fact that
\begin{equation}
e^x-x-1\leq \frac{3}{4}x^2 \label{eq:e-to-x-quad-bound}    
\end{equation}
for all $x\in[0, 1]$.
\end{proof}

\begin{remark}
Eq.~\eqref{tmp:c3} and Lemma~\ref{lem:main-upper-lemma} represent the most important step in the proof above, which leads to our improvement over the prior bound from \cite{go2025SamplebasedHamiltonianLindbladian}.
\end{remark}

\begin{lemma}\label{lem:main-upper-lemma}
Let $d\in\mathbb{N}, d\geq2$ and let $\mathcal{D}_1[\cdot] \coloneq \frac{\mathbb{I}_d}{d}\operatorname{Tr}_1[\cdot]$ be the completely depolarizing channel acting on a $d$-dimensional quantum system labeled as $1$. Letting $\Delta\geq0$, there exists a linear map $\widetilde{\mathcal{L}}$ corresponding to the quantum channel $\widetilde{e^{\mathcal{L}\Delta}}$ (defined in~\eqref{eq:WML-actual-one-step}) such that
\begin{equation}
    \widetilde{e^{\mathcal{L}\Delta}}=\mathcal{I} + \Delta\left(\mathcal{I}-\mathcal{D}_1\right)\circ\widetilde{\mathcal{L}}
\end{equation}
with
\begin{equation}
    \left\|\widetilde{\mathcal{L}}-\mathcal{L}\right\|_\diamond\leq \frac{\Delta d}{4}\left\Vert L\right\Vert_\infty^2.
\end{equation}
\end{lemma}

\begin{proof}
    See Appendix~\ref{appendix:main-upper-proof}.
\end{proof}

\medskip

The main intuition behind the proof is that the WML jump operator $M$ in~\eqref{eq:WML-M-op} has a repetitive structure, i.e., $MM^\dagger M = dM$ and $\sqrt{d}M^2=M$. Thus, higher-order terms can be represented as quadratic polynomials of $M$ and $M^\dagger$. A direct consequence of Lemma~\ref{lem:main-upper-lemma} is the following linear approximation to $\mathcal{L}$:
\begin{align}
    \left\Vert \widetilde{e^{\mathcal{L}\Delta}} - \mathcal{I} - \mathcal{L}\Delta\right\Vert_\diamond &=\left\Vert
    \Delta\left(\mathcal{I}-\mathcal{D}_1\right)\circ\widetilde{\mathcal{L}} - \Delta 
     \mathcal{L}  + \Delta \mathcal{D}_1 \circ \mathcal{L} \right\Vert_\diamond \notag \\
    &=\Delta\left\Vert(\mathcal{I} - \mathcal{D}_1)\circ (\widetilde{\mathcal{L}} - \mathcal{L})\right\Vert_\diamond \notag \\
    &\leq\Delta\left\Vert\widetilde{\mathcal{L}} - \mathcal{L}\right\Vert_\diamond + \Delta\left\Vert\mathcal{D}_1\circ (\widetilde{\mathcal{L}} - \mathcal{L})\right\Vert_\diamond \notag \\
    &\leq 2\Delta\left\Vert\widetilde{\mathcal{L}} - \mathcal{L}\right\Vert_\diamond \notag \\
    &\leq \frac{\Delta^2d}{2}\left\Vert L\right\Vert_\infty^2.
\end{align}
To establish the first equality above, we used the fact that $\mathcal{D}_1 \circ \mathcal{L}=0$ because $\operatorname{Tr}[(\mathcal{D}_1 \circ \mathcal{L})(\rho)] = \operatorname{Tr}[\mathcal{L}(\rho)]=0$ for every state $\rho$ (i.e., in the latter we used that $\mathcal{D}_1$ is trace preserving and $\mathcal{L}$ is trace annihilating). In the following sections, we employ tools from random matrix theory to establish an upper bound on the typical sample complexity in Def.~\ref{def: approx sample complexity} that does not increase as the dimension grows.

\section{Typical-case scaling for random Lindbladian} \label{sec:average-case}
The sample complexity in Definition~\ref{def: sample complexity} requires the simulation error to be bounded from above by $\varepsilon$ for every possible quantum channel $e^{\mathcal{L}_{\pi}t}$ specified by $\pi$. Although such a worst-case guarantee is operationally strong, it may be overly restrictive in practice, because, in reality, physical noise models are rarely adversarial. In many situations, dissipative channels arise from complex environments whose effective descriptions resemble random operators with finite moments~\cite{rigol2008ThermalizationItsMechanism, bruzda2009RandomQuantumOperations, denisov2019UniversalSpectraRandom, kukulski2021GeneratingRandomQuantum, sa2020SpectralSteadystateProperties, chruscinski2025SpectralDelineationMarkov, moske2025RandomMatrixPerspective}. This motivates us to investigate typical behavior of the sample complexity when Lindbladian operators are drawn from random distributions. We therefore introduce a \emph{relaxed} notion of sample complexity for quantum simulation, which requires the simulation error to be bounded with high probability over the choice of program state. 
\begin{definition}[Typical sample complexity of sample-based quantum simulation]
\label{def: approx sample complexity}
The typical sample complexity of sample-based simulation $\overline{n}_d(t,\varepsilon,\delta)$ is defined as the minimum number of program states needed to realize a channel  $e^{\mathcal{L}_{\pi}t}$ within error $\varepsilon$, with a probability of at least $1-\delta$ over a uniformly random choice of the program state $\pi$:
\begin{equation}
    \overline{n}_d(t,\varepsilon,\delta) \coloneqq \inf_{\mathcal{P}^{(n)}\in\rm{CPTP}
    }\left\{  n\in\mathbb{N}:\Pr_{\pi\sim\mathcal{U}(\mathcal{H}_d)} \left[ \frac{1}{2}\left\Vert \mathcal{P}^{(n)}
    \circ\mathcal{A}_{\pi^{\otimes n}}-e^{\mathcal{L}_{\pi}t}\right\Vert
    _{\diamond}\leq\varepsilon \right] \geq 1-\delta \right\},
\end{equation}
where $\mathcal{U}(\mathcal{H}_d)$ denotes a uniform distribution on $\mathcal{D}(\mathcal{H})$.
\end{definition}

We note that the sample complexity in Def.~\ref{def: sample complexity} and the typical sample complexity in Def.~\ref{def: approx sample complexity} are related as
\begin{equation}
    n^{\ast}_d(t,\varepsilon)\geq \overline{n}_d(t,\varepsilon,\delta),
\end{equation}
where the equality holds by taking $\delta = 0$. Throughout this paper, we will omit $\delta$ if the implicit meaning is clear; i.e., we will abuse notation for typical sample complexity as $\overline{n}_d(t,\varepsilon)$ if typicality is implied in the context.

Following Defs.~\ref{def: sample complexity} and~\ref{def: approx sample complexity}, we consider (normalized) dynamics parameterized by a program state. Since the matrix elements of the program state are those of the Kossakowski matrix that fully describes Lindbladian dynamics (see Appendix~\ref{appendix:state2dyn} for details), we can consider a distance measure on program states, instead of that on normalized Lindbladian generators. This kind of randomness, based on the state-to-dynamics relation, is common when addressing random quantum channels~\cite{bruzda2009RandomQuantumOperations, nechita2018AlmostAllQuantum, denisov2019UniversalSpectraRandom, kukulski2021GeneratingRandomQuantum}. We therefore take the Haar-random measure as we focus on pure program states; i.e., $\ket{L}$ is a Haar-random pure state vector on $\mathbb{C}^d\otimes\mathbb{C}^d$. Equivalently, $|L\rangle$ may be generated by drawing a random vector from a rotation-invariant ensemble, commonly chosen to be a vector with independent complex Gaussian entries, followed by normalization~\cite{watrous2018TheoryQuantumInformation, zyczkowski2001InducedMeasuresSpace, kukulski2021GeneratingRandomQuantum}. Therefore, sampling $L$ from the Ginibre ensemble and normalizing it is equivalent to sampling the program state vector $\ket{L}$ uniformly and identifying the coefficients of $\ket{L}$ in the product basis with the matrix entries of $L$. For general dissipative Lindbladians, we can consider the induced measure on mixed states for program states~\cite{zyczkowski2001InducedMeasuresSpace}.

Theorem~\ref{thm:average-case-WML} establishes a high-probability upper bound on the sample complexity. Its proof relies on a concentration bound for the normalized operator norm of a Ginibre ensemble, as given in Proposition~\ref{prop:concentration}.

\begin{theorem}[Typical-case sample complexity of WML]
\label{thm:average-case-WML}
Let $d\ge 2$, and let $G \in \mathbb{C}^{d \times d}$ be a complex Ginibre matrix whose entries are independent complex Gaussian random variables with mean zero and unit variance, and define $L \coloneqq \frac{G}{\left\|G\right\|_2}$.
Then, for all $t \ge 0$, $\varepsilon \in (0,1)$, and $\delta \in (0,1)$, the typical case sample complexity given by Def.~\ref{def: approx sample complexity} satisfies
\begin{equation}
    \overline{n}_d(t,\varepsilon,\delta) \le
    7 \left( 1+\frac{\log(2/\delta)}{d} \right) \frac{t^2}{\varepsilon}.
\end{equation}
In particular, as the system dimension $d$ increases, the correction term $\log(2/\delta)/d$ vanishes asymptotically, and the sample complexity scales as
\begin{equation}
    \overline{n}_d(t,\varepsilon,\delta) = O\!\left(\frac{t^2}{\varepsilon}\right).
\end{equation}
\end{theorem}

\begin{proof}
By Proposition~\ref{prop:concentration}, the Frobenius-normalized Ginibre operator $L=G/\left\|G\right\|_2$ satisfies, with probability at least $1-\delta$,
\begin{equation}
    \left\|L\right\|_\infty^2 = \frac{\left\|G\right\|_\infty^2}{\left\|G\right\|_2^2}
    \le \frac{16}{d} + \frac{8\log(2/\delta)}{d^2}.
\end{equation}
Conditioned on this event, we apply the WML sample-complexity upper bound in Theorem~\ref{cor:samp-comp-upper-bnd-WML}. 
For a fixed Lindblad operator $L$, the number of WML samples sufficient to achieve diamond-distance error at most $\varepsilon$ is bounded by
\begin{equation}
    n_d(t,\varepsilon) \le \frac{2d+3}{8}
    \left\|L\right\|_\infty^2 \frac{t^2}{\varepsilon}.
\end{equation}
Substituting the high-probability bound on $\left\|L\right\|_\infty^2$ gives
\begin{align}
    n_d(t,\varepsilon)
    &\le \frac{2d+3}{8} \left(\frac{16}{d} + \frac{8\log(2/\delta)}{d^2}\right)\frac{t^2}{\varepsilon} \\
    & = \left[4+\frac{6}{d}+\left(\frac{2}{d}+\frac{3}{d^2}\right)\log(2/\delta)\right]\frac{t^2}{\varepsilon}.
\end{align}
For $d\ge 2$, this is bounded from above by
\begin{equation}
    n_d(t,\varepsilon)\le 7\left(1+\frac{\log(2/\delta)}{d}\right)\frac{t^2}{\varepsilon}.
\end{equation}
By Def.~\ref{def: approx sample complexity}, this implies
\begin{equation}
    \overline{n}_d(t,\varepsilon,\delta)\le 7\left(1+\frac{\log(2/\delta)}{d}\right)\frac{t^2}{\varepsilon}.
\end{equation}
By noting that the correction term $\log(2/\delta)/d$ vanishes asymptotically as $d$ grows, we obtain
\begin{equation}
    \overline{n}_d(t,\varepsilon,\delta) = O\!\left(\frac{t^2}{\varepsilon}\right),
\end{equation}
which completes the proof.
\end{proof}

\begin{proposition}[Concentration of the normalized operator norm]
\label{prop:concentration}
Let $d\ge 2$, and let $G \in \mathbb{C}^{d \times d}$ be a complex Ginibre matrix whose entries are independent complex Gaussian random variables with mean zero and unit variance, and define $L \coloneqq \frac{G}{\left\|G\right\|_2}$.
Then, for all $\delta \in (0,1)$, with probability at least $1-\delta$,
\begin{equation}
    \left\|L\right\|_\infty^2 \le \frac{16}{d} + \frac{8\log(2/\delta)}{d^2}.
\end{equation}
\end{proposition}

\begin{proof}
    We estimate the normalized operator norm of the Frobenius-normalized Ginibre operator $L=G/\left\|G\right\|_2$.
    Since $\left\|L\right\|_\infty^2 =\left\|G\right\|_\infty^2/\left\|G\right\|_2^2$, it suffices to control the operator norm of $G$ from above and the Frobenius norm of $G$ from below.

    For the numerator, we use the one-sided Gaussian concentration inequality for the operator norm of a complex Ginibre matrix. 
    Specifically, for every $u\ge0$,
    \begin{equation}
        \mathbb{P}\!\left( \left\|G\right\|_\infty > 2\sqrt d+u \right)
        \le \exp\!\left(-\frac{u^2}{2}\right)   ,
    \end{equation}
    which follows from the Gaussian concentration argument underlying Ref.~\cite[Corollary~7.3.2]{vershynin2018HighDimensionalProbabilityIntroduction}.

    For the denominator, we use that $\left\|G\right\|_2^2=\sum_{i,j}|G_{ij}|^2$. Since $G_{ij}\sim\mathcal{CN}(0,1)$, the random variables $|G_{ij}|^2$ are independent exponential random variables with mean one. Hence $2\left\|G\right\|_2^2$ is $\chi$-squared with $2d^2$ degrees of freedom. By the Laurent--Massart $\chi$-squared inequality~\cite[Lemma~1, Eq.~(4.4)]{laurent2000AdaptiveEstimationQuadratic}, for every $x>0$,
    \begin{equation}
        \mathbb{P}\!\left(\left\|G\right\|_2^2 < d^2-d\sqrt{2x}\right)\le e^{-x}.
    \end{equation}
    Equivalently, for $\eta\in(0,1)$,
    \begin{equation}
        \mathbb{P}\!\left(\left\|G\right\|_2^2 < (1-\eta)d^2\right)\le \exp\!\left(-\frac{\eta^2d^2}{2}\right).
    \end{equation}
    Now let us define
    \begin{equation}
        X(u,\eta)\coloneqq\frac{(2\sqrt d+u)^2}{(1-\eta)d^2}.
    \end{equation}
    By the preceding two tail bounds and the union bound, for any $u\ge0$ and $\eta\in(0,1)$,
    \begin{align}
        \mathbb{P}\!\left(\frac{\left\|G\right\|_\infty^2}{\left\|G\right\|_2^2} > X(u,\eta)\right)
        &\leq \mathbb{P}\!\left(\left\|G\right\|_\infty > 2\sqrt d+u\right) 
            +\mathbb{P}\!\left(\left\|G\right\|_2^2 < (1-\eta)d^2\right)\nonumber\\
        &\le \exp\!\left(-\frac{u^2}{2}\right)+\exp\!\left(-\frac{\eta^2 d^2}{2}\right).
    \end{align}
    For a given $\delta\in(0,1)$, we choose
    \begin{equation}
        u = \sqrt{2\log(2/\delta)}, \qquad
        \eta = \frac{\sqrt{2\log(2/\delta)}}{d}.
    \end{equation}
    Each exponential term is then bounded by $\delta/2$, which leads to 
    \begin{equation}
        \mathbb{P}\!\left(\frac{\left\|G\right\|_\infty^2}{\left\|G\right\|_2^2} > X(u,\eta)\right) \le \delta.
    \end{equation}
    Suppose first that $\eta\le 1/2$. 
    Then $(1-\eta)^{-1}\le2$, and therefore
    \begin{align}    
        X(u,\eta) &\le \frac{2\left(2\sqrt d+\sqrt{2\log(2/\delta)}\right)^2}{d^2}\\
        &\le \frac{2\left(8d+4\log(2/\delta)\right)}{d^2}\\
        &= \frac{16}{d}+ \frac{8\log(2/\delta)}{d^2},
    \end{align}
    where we used $(a+b)^2\le 2a^2+2b^2$ in the second inequality.
    Consequently, with probability at least $1-\delta$,
    \begin{equation}
        \left\|L\right\|_\infty^2 = \frac{\left\|G\right\|_\infty^2}{\left\|G\right\|_2^2}
        \le \frac{16}{d} + \frac{8\log(2/\delta)}{d^2}.
    \end{equation}
    This completes the proof.
\end{proof}

\medskip

A similar dimension-independent typical-case scaling also holds for more general subgaussian ensembles, up to constants depending only on the subgaussian norm bound. 
We include this extension in Appendix~\ref{app:subgaussian-average-case}.

\section{Worst-case scaling under adversarial Lindbladians}
\label{sec:worst-case}
Lemma~\ref{Thm: error bound WML} gives a diamond-distance error bound for WML whose leading prefactor is proportional to $d\left\|L\right\|_\infty^2$. In Section~\ref{sec:average-case}, we showed that this dimension dependence is canceled with high probability for Frobenius-normalized Ginibre Lindblad operators, since $\left\|L\right\|_\infty^2=O(1/d)$. It is therefore natural to investigate the worst-case behavior of WML, where the operator-norm contribution is of constant order. A canonical way to realize this regime is to consider low-rank Lindblad operators because $\left\|L\right\|_\infty^2=\Theta(1)$ for a Frobenius-normalized rank-one operator $L$. 

We establish the following dimension-dependent lower bound on the WML sample complexity for this rank-one Lindblad operator, which exhibits linear scaling with the dimension.
\begin{theorem}[Lower bound for WML on an adversarial Lindbladian]
\label{prop:rankone_lower}
Let us define the sample complexity of the WML algorithm as
\begin{equation}
    n_{\mathrm{WML}}(t,\varepsilon) \coloneqq 
    \inf\left\{n\in\mathbb N: \sup_{\mathcal{L}\in\mathcal{B}(\mathcal{H}_d)}\frac{1}{2}\left\|
    \widetilde{e^{\mathcal{L} t}} - e^{\mathcal{L} t}\right\|_\diamond
    \le \varepsilon
    \right\},
\end{equation}
for $d\geq 2$ and all $t,\varepsilon>0$, 
Then, for sufficiently small $\varepsilon$, the lower bound of the sample complexity of WML is given as
\begin{equation}
n_{\mathrm{WML}}(t,\varepsilon)
\ge
\frac{d}{32}\,
\frac{t^2}{\varepsilon},
\label{eq:rankone_lower_rzero}
\end{equation}
which can be obtained from an adversarial Lindbladian $\mathcal{L}$ with a rank-one jump operator $L\coloneqq |u\rangle\langle u|$, where $|u\rangle\in\mathbb C^d$ and $\left\| |u\rangle \right\|_2=1$.
\end{theorem}

\begin{proof}
See Appendix~\ref{app:worst_prop_proof}.
\end{proof}

\section{Implications and applications} \label{sec:implication}
In this section, we explain protection against tomography of program states in sample-based simulation, related to what has been discussed previously in the context of density matrix exponentiation~\cite{kimmel2017HamiltonianSimulationOptimal} and WML \cite{patel2023WaveMatrixLindbladization, patel2023WaveMatrixLindbladizationa}. As an alternative method to sample-based Lindbladian simulation, quantum state tomography can be performed on i.i.d.~copies of program states, in order to obtain an estimate $\widetilde{\pi}$ of $\pi$, or equivalently, an estimate $\widetilde L$ of $L$. We recall that the dimension of the program state $\pi$ is at most $d^2$ for $d$-dimensional Lindbladian dynamics (see Appendix~\ref{appendix:state2dyn}). Therefore, the sample complexity of program state tomography should be lower bounded by $\Omega(d^2r)$~\cite{haah2016SampleoptimalTomographyQuantum, haah2017SampleoptimalTomographyQuantum, yuen2023ImprovedSampleComplexity}. More specifically, Ref.~\cite[Eq. (2.84)]{patel2023WaveMatrixLindbladizationa} proved that 
\begin{equation}
    n=\Omega\left(\frac{d^2(t-\varepsilon)^2}{\varepsilon^2\log(d^2t/\varepsilon)}\right)
\end{equation}
samples are necessary to achieve $\frac{1}{2}\left\Vert e^{\mathcal{L}t} - e^{\widetilde{\mathcal{L}}t}\right\Vert_\diamond\leq O(\varepsilon)$, where $e^{\widetilde{\mathcal{L}}t}$ is a channel that approximates the dynamics based on the jump operator $\tilde L$ estimated from program state tomography. However, the sample complexity of sample-based Lindbladian simulation is typically upper bounded by $O(t^2/\varepsilon)$, and $O(dt^2/\varepsilon)$ even for the worst case. Therefore, for a sufficiently large $d$, this strict gap between the Lindbladian-simulation upper bound and the program-state-tomography lower bound provides a form of quantum privacy~\cite{aaronson2009QuantumCopyProtectionQuantum, patel2023WaveMatrixLindbladizationa}. This clear distinction between the sample complexities of simulation and tomography was not evident prior to this work, since the previously best-known sample complexity upper bound for Lindbladian simulation, $O(d^2t^2/\varepsilon)$~\cite{go2025SamplebasedHamiltonianLindbladian}, scales quadratically with the dimension.

\section{Concluding remarks} \label{sec:conclusion}
We established improved upper bounds on the sample complexity of sample-based Lindbladian simulation based on Wave Matrix Lindbladization (WML). Our main contribution was to improve the dimension dependence of the sample complexity, yielding the upper bound $n^\ast_d(t, \varepsilon) = O(d \| L\|_\infty^2 t^2/\varepsilon)$. Our result tightens the gap between previously known upper and lower bounds while preserving the optimal $O(t^2/\varepsilon)$ scaling inferred from a sample-based simulation framework. We further demonstrated that the dimension dependence of the sample complexity vanishes in the typical case for random Lindblad operators. In this case, we obtain the sample complexity $O(t^2/\varepsilon)$ that does not increase by the dimension, with high probability $1-\delta$. This provides rigorous evidence that realistic noise models are typically less complex than worst-case adversarial instances and possess structural features that enable WML to achieve substantially lower sample complexity. On the other hand, we also identified an explicit case in which linear scaling of the sample complexity is unavoidable for WML, thereby revealing a strict gap between the worst- and typical-case complexity of the algorithm. Our findings strengthen the theoretical foundation of sample-based open-system simulation by providing a rigorous analysis of sample complexity and opening a new research direction toward exploring the complexity dichotomy between typical and adversarial Lindbladian dynamics.

Several directions remain open. First, it would be valuable to extend the present analysis to general Lindbladians with multiple dissipative terms and coherent Hamiltonian contributions. Second, it is natural to ask whether our newly obtained upper bounds, in the typical and worst case, match the corresponding sample-complexity lower bound. Currently available analysis on the lower bound is restricted to $d=2$~\cite[Theorem~5]{go2025SamplebasedHamiltonianLindbladian}.
We conjecture that the remaining $O(d)$ multiplicative gap between the WML-based upper bound and the algorithm-independent lower bound is irreducible. The intuition is that, as the operator norm of the Lindblad operator decreases, different dissipative channels become harder to distinguish, while fewer samples are required to achieve the same level of channel error. Since sample-complexity lower bounds are closely related to channel distinguishability~\cite{wei2024SimulatingNoncompletelyPositive, go2025SamplebasedHamiltonianLindbladian}, this tradeoff suggests that the $O(d)$ gap may be intrinsic. Finally, further investigation of experimentally motivated noise models may help connect these complexity-theoretic results to practical simulation strategies for near-term quantum devices, including continuous-variable quantum systems.

\begin{acknowledgments}
This work was supported by the National Research Foundation of Korea (NRF) Grants (No.~RS-2024-00413957, No.~RS-2024-00438415). HK is supported by the KIAS Individual Grant No.~CG085302 at Korea Institute for Advanced Study. DP and MMW acknowledge support from the Air Force Office of Scientific Research under agreement no.~FA2386-24-1-4069.    
    
The U.S.~Government is authorized to reproduce and distribute reprints for Governmental purposes notwithstanding any copyright notation thereon. The views and conclusions contained herein are those of the authors and should not be interpreted as necessarily representing the official policies or endorsements, either expressed or implied, of the United States Air Force.

\end{acknowledgments}

\section*{Author contributions}

\noindent \textbf{Author Contributions}:
The following describes the different contributions of the authors of this work, using roles defined by the CRediT
(Contributor Roles Taxonomy) project~\cite{NISO}:

\medskip 
\noindent \textbf{SP}: FFormal analysis, Validation, Writing - Original draft, Writing—Review \& Editing.

\medskip 
\noindent \textbf{YS}: Formal analysis, Validation, Writing - Original draft, Writing—Review \& Editing.

\medskip 
\noindent \textbf{BG}: Validation, Writing - Review \& Editing.

\medskip 
\noindent \textbf{DP}: Validation, Writing - Review \& Editing.

\medskip 
\noindent \textbf{MMW}: Funding acquisition,  Validation,  Writing - Review \& Editing.

\medskip 
\noindent \textbf{HK}: Formal analysis, Methodology, Funding acquisition, Writing - Review \& Editing, Supervision.

\bibliography{references}

\appendix

\section{State to dynamics correspondence in sample-based Lindbladian simulation}

\label{appendix:state2dyn}

In this section, we explain why program state for dissipative Lindbladian dynamics needs to be $d^2$ dimensional quantum state. The general GKSL form is composed of coherent and dissipative parts~\cite{gorini1976CompletelyPositiveDynamical, lindblad1976GeneratorsQuantumDynamical},
\begin{equation}
    \mathcal{L} = \mathcal{L}_U + \mathcal{L}_D,
\end{equation}
where the coherent part is
\begin{equation}
    \mathcal{L}_U[\rho] = -i[H, \rho],
\end{equation}
and the dissipative part is
\begin{equation}
    \mathcal{L}_D[\rho] = \sum_{n,m=1}^{d^2-1}K_{mn}\left[F_n\rho F_m^\dagger - \frac{1}{2}\{F_m^\dagger F_n, \rho\}\right].
\end{equation}
Here, $H\in\mathfrak{su}(d)$, $\{F_n\}_{n=1}^{d^2-1}$ is an arbitrary (Hermitian) Hilbert-Schmidt basis on $\mathfrak{su}(d)$ such that $\operatorname{Tr}(F_nF_m^\dagger)=\delta_{mn}$, and $K$ is a positive semi-definite matrix known as the ``Kossakowski matrix," or the GKS matrix. $\mathcal{L}$ becomes Lindbladian if and only if there exists a positive semi-definite matrix $K$ and a Hermitian matrix $H$~\cite{sweke2015UniversalSimulationMarkovian, childs2023EfficientSimulationSparse, denisov2019UniversalSpectraRandom}. We can partially integrate the coherent parts with the dissipative parts in a unified form as
\begin{equation}
    \mathcal{L}[\rho]\coloneq-i[H',\rho]+\sum_{m,n=0}^{d^2-1} \Pi_{mn} \left[F_n\rho F_m^\dagger-\frac{1}{2}\left\{F_m^\dagger F_n, \rho\right\}\right]
\end{equation}
with a coefficient matrix
\begin{equation}
    \Pi=\begin{bmatrix}
        \Pi_{00} & -ih^T \\
        ih & K
    \end{bmatrix}\succeq0.
\end{equation}
Here, $F_0$ commutes with any $F_n$, e.g., $F_0=\tfrac{\mathbb{I}_d}{\sqrt{d}}$, and $h$ is a $(d^2-1)$-dimensional real vector. The positive semi-definiteness condition implies that $\Pi_{00}\geq h^\dagger K^{-1}h$. This suggests that there exists a quantum \emph{program} state
\begin{equation}\label{eq:progdef}
    \pi\coloneq\frac{1}{\operatorname{Tr}\Pi}\sum_{m,n=0}^{d^2-1}{\Pi_{mn}}\ket{F_n}\bra{F_m}\in\mathcal{D}(\mathcal{H}_{d\times d})
\end{equation}
with complete information of a dissipative Lindblad generator up to normalization. This motivates to sample-based Lindbladian simulation, as sample-based Hamiltonian simulation~\cite{lloyd2014QuantumPrincipalComponent, kimmel2017HamiltonianSimulationOptimal, wei2024SimulatingNoncompletelyPositive, go2025SamplebasedHamiltonianLindbladian} utilize $\sigma\coloneq \frac{H-\sigma_0 \mathbb{I}}{\operatorname{Tr}[H-\sigma_0 \mathbb{I}]}$ as a program state. Non-scalar residual coherent term, $H'\coloneq H-\sum_mh_mF_mF_0$, should be considered when $\Pi_{00}$ cannot be bounded, e.g., with separate sample-based Hamiltonian simulation and Trotterization techniques~\cite{patel2023WaveMatrixLindbladizationa}. Otherwise, we can set $H'=0$.

By diagonalizing $\Pi$, we recover a dissipative Lindbladian as linear combination of \emph{single} dissipative terms $\mathcal{L}^{(1)}$, where 
\begin{equation}
    \label{eq:lindblad_dynamics_app}
    \mathcal{L}^{(1)}_k(\rho) \coloneqq L_k \rho L_k^\dagger - \frac{1}{2} \left\{L_k^\dagger L_k, \rho \right\}, 
\end{equation}
and $\Pi_{mn}=\sum_k\gamma_ku^\ast_{mk}u_{nk}$, $L_k=\sum_m u_{mk}F_m$.
The corresponding program states then become an ensemble of pure states
\begin{equation}
    \pi_{L_k}=\ket{L_k}\bra{L_k}
\end{equation}
with distribution $\{\gamma_k\}$. If a sample-based simulation of a single dissipator is achievable, then simulation of general Lindbladian dynamics is also possible, e.g., via linear combinations of unitary operations (LCU)~\cite[Algorithm 1]{patel2023WaveMatrixLindbladizationa}. Although the extension to general Lindbladian simulation is also feasible via direct preparation of the mixed state $\pi = \sum_k \gamma_k \pi_{L_k}$, simulating coherent and dissipative terms simultaneously may be inefficient because $\Pi_{00}$ may diverge to maintain the positive semi-definite property of the program state. From the above discussion, we focus on a single dissipative Lindbladian, without loss of generality.

\section{Proof of Lemma~\ref{lem:main-upper-lemma}} \label{appendix:main-upper-proof}
As mentioned in the main text, higher-order terms in $(M, M^\dagger)$ appearing in the exponential map can be reduced to obtain an analytic representation of the WML Lindbladian. Formally, the operator $M$ in~\eqref{eq:WML-M-jump} satisfies the following relations:
\begin{align}
M^2 & =\frac{1}{\sqrt{d}}M,\\
M^\dagger MM^\dagger & = dM^\dagger,
\end{align}
which end up eliminating third orders of $M, M^\dagger$. 

Let us define the following five linear maps acting on an arbitrary quantum state $\rho$:
\begin{align}
    \mathcal{S}_0(\rho) &\coloneqq  \rho,\\
    \mathcal{S}_1(\rho) &\coloneqq  M\rho M^\dagger,\\
    \mathcal{S}_2(\rho) &\coloneqq  \frac{1}{2}\{M^\dagger M, \rho\},\\
    \mathcal{S}_3(\rho) &\coloneqq  \frac{\sqrt{d}}{2}\left(M^\dagger M\rho M^\dagger + M\rho M^\dagger M\right),\\
    \mathcal{S}_4(\rho) &\coloneqq  \frac{1}{d} M^\dagger M\rho M^\dagger M.
\end{align}
The coefficients of these maps have been selected by normalization, i.e., $\left\|\operatorname{Tr}_{23}\mathcal{S}_{i}(~\cdot~\otimes\pi_L)\right\|_\diamond=1$ for all $i$ and for normalized $L$, following from~\cite{go2025SamplebasedHamiltonianLindbladian}. One can immediately figure out that
\begin{align}
\mathcal{S}_0 & =\mathcal{I},\\
\mathcal{M} & =\mathcal{S}_1-\mathcal{S}_2,
\end{align}
from~\eqref{eq:WML-M-op}, and also that 
\begin{align}
\operatorname{Tr}_{23}\mathcal{S}_{0}(\rho\otimes\pi_L) & =\rho,
\\
\operatorname{Tr}_{23}\mathcal{S}_{1}(\rho\otimes\pi_L) & =L\rho L^\dagger, \\
\operatorname{Tr}_{23}\mathcal{S}_{2}(\rho\otimes\pi_L) & =\tfrac{1}{2}\{L^\dagger L, \rho\},
\end{align}
from~\cite{patel2023WaveMatrixLindbladization}. Crucially, $\mathcal{M}$ acting on a linear combination of the $\mathcal{S}_i$ superoperators produces  a linear combination of them, as follows:
\begin{align}
    \mathcal{M}\mathcal{S}_0 &= \mathcal{S}_1-\mathcal{S}_2,\\
    \mathcal{M}\mathcal{S}_1 &= \frac{1}{d}\mathcal{S}_1-\frac{1}{d}\mathcal{S}_3,\\
    \mathcal{M}\mathcal{S}_2 &= d\mathcal{S}_1-\frac{d}{2}\mathcal{S}_2-\frac{d}{2}\mathcal{S}_4,\\
    \mathcal{M}\mathcal{S}_3 &= d\mathcal{S}_1-\frac{d}{2}\mathcal{S}_3-\frac{d}{2}\mathcal{S}_4,\\
    \mathcal{M}\mathcal{S}_4 &= d\mathcal{S}_1-d\mathcal{S}_4.
\end{align}
We can simplify the relations above by embedding them in a transfer matrix $T$, such that
\begin{equation}
\mathcal{M}\mathcal{S}_{k} = \sum_j T_{kj} \,\mathcal{S}_{j}  ,
\end{equation}
and the matrix elements of $T$ are as follows:
\begin{equation}
    T
    =
    \begin{bmatrix}
        0 & 1 & -1 & 0 & 0 \\
        0 & \frac{1}{d} & 0 & -\frac{1}{d} & 0 \\
        0 & d & -\frac{d}{2} & 0 & -\frac{d}{2} \\
        0 & d & 0 & -\frac{d}{2} & -\frac{d}{2} \\
        0 & d & 0 & 0 & -d
    \end{bmatrix}.
\end{equation}
Due to linearity, computing the exponentiation of $T$ suffices for representing the exponential map $e^{\mathcal{M}\Delta}$ acting on the $\mathcal{S}_i$ superoperators, in particular, on the identity map $\mathcal{S}_0$:
\begin{equation}
    e^{\mathcal{M}\Delta} = \sum_j (e^{T\Delta})_{0j}\,\mathcal{S}_j \, .
\end{equation}
The remaining task is to exponentiate the five-dimensional matrix $T$ using its eigendecomposition. Here, we skip the detailed process and write the relevant result as follows:
\begin{align}
    (e^{T\Delta})_{00} &= 1, \\
    (e^{T\Delta})_{01} &=-\frac{e^{-\Delta d'}-1}{d'}, \\
    (e^{T\Delta})_{02} &=\frac{e^{-\Delta d/2}-1}{d/2}, \\
    (e^{T\Delta})_{03} &=\frac{2}{d^2-2}\left((e^{T\Delta})_{01}+(e^{T\Delta})_{02}\right), \\
    (e^{T\Delta})_{04} &=-\frac{d^2}{d^2-2}\left((e^{T\Delta})_{01}+(e^{T\Delta})_{02}\right),
\end{align}
where we define $d'\coloneqq d-\tfrac{1}{d}$. The analytic expression for $\operatorname{Tr}_{23}[e^{\mathcal{M}\Delta}(\rho\otimes\pi_L)]$ then follows directly as,
\begin{align}
    \widetilde{e^{\mathcal{L}\Delta}}(\rho_{01})
    &=\rho_{01}  +(e^{T\Delta})_{02}\frac{1}{2}\left\{L^\dagger L, \rho_{01}\right\} \label{eq:tmp1} \\
    &\qquad +\left(\frac{2}{d^2-2}(e^{T\Delta})_{02}+\frac{d^2}{d^2-2}(e^{T\Delta})_{01}\right) L\rho_{01} L^\dagger \label{eq:tmp2} \\
    &\qquad - \frac{d^2}{d^2-2}\left((e^{T\Delta})_{02} + (e^{T\Delta})_{01}\right) \mathcal{D}_1[L\rho_{01}L^\dagger] ,\label{eq:tmp3}
\end{align}
where $\mathcal{D}_1=\frac{\mathbb{I}}{d}\operatorname{Tr}_1$ is the completely depolarizing channel. The coefficient in~\eqref{eq:tmp2} follows from the fact that $\operatorname{Tr}_{23}[\mathcal{S}_3(\rho\otimes\pi_L)]=L\rho L^\dagger$, and the mapping of~\eqref{eq:tmp3} follows from the identity $\operatorname{Tr}_{23}[\mathcal{S}_4(\rho\otimes\pi_L)]=\tfrac{\mathbb{I}}{d}\operatorname{Tr}_1[L\rho L^\dagger]$. Substituting coefficients gives
\begin{align}
    \widetilde{e^{\mathcal{L}\Delta}}(\rho_{01})
    &=\rho_{01}+\Delta\mathcal{L}(\rho_{01}) \notag \\
    &\qquad +\Delta\left(1+\frac{e^{-\Delta d/2}-1}{\Delta d/2}\right)\frac{1}{2}\{L^\dagger L, \rho_{01} \} \notag \\
    &\qquad +\Delta\left(\frac{2}{d^2-2}\left[1+\frac{e^{-\Delta d/2}-1}{\Delta d/2}\right]-\frac{d^2}{d^2-2}\left[1+\frac{e^{-\Delta d'}-1}{\Delta d'}\right]\right) L\rho_{01} L^\dagger \notag \\
    &\qquad +\Delta \frac{d^2}{d^2-2}\left(1+\frac{e^{-\Delta d'}-1}{\Delta d'}-1-\frac{e^{-\Delta d/2}-1}{\Delta d/2}\right) \mathcal{D}_1[L\rho_{01}L^\dagger].
\end{align}
We note that $\mathcal{D}_1$ acting on the terms in~\eqref{eq:tmp1} and~\eqref{eq:tmp2} is exactly~\eqref{eq:tmp3} with the opposite sign, from the cyclic property of trace: $2\mathcal{D}_1[L\rho L^\dagger] = \mathcal{D}_1[\{L^\dagger L,\rho\}]$. Therefore, letting $\widetilde{\mathcal{L}}$ be the linear map such that
\begin{align}
    \widetilde{\mathcal{L}}[\rho_{01}]&= \mathcal{L}[\rho_{01}] + \left(1+\frac{e^{-\Delta d/2}-1}{\Delta d/2}\right)\frac{1}{2}\{L^\dagger L, \rho_{01} \} \notag \\
    &\qquad+\left(\frac{2}{d^2-2}\left[1+\frac{e^{-\Delta d/2}-1}{\Delta d/2}\right]-\frac{d^2}{d^2-2}\left[1+\frac{e^{-\Delta d'}-1}{\Delta d'}\right]\right) L\rho_{01} L^\dagger,
\end{align}
we obtain
\begin{equation}
    \widetilde{e^{\mathcal{L}\Delta}}=\mathcal{I} + \Delta(\mathcal{I}-\mathcal{D}_1)\widetilde{\mathcal{L}},
\end{equation}
where we used the fact that $\mathcal{D}_1 \circ \mathcal{L} = 0$.
This confirms the first part of Lemma~\ref{lem:main-upper-lemma}.

The second part follows from the fact that $0\leq f(x)\coloneqq 1+\frac{e^{-x}-1}{x}\leq\tfrac{x}{2}$ which monotonically increases for all $ x\in[0, \infty)$:
\begin{align}
    & \left\Vert\widetilde{\mathcal{L}}-\mathcal{L}\right\Vert_\diamond \notag \\
    &\leq \left\vert1+\frac{e^{-\Delta d/2}-1}{\Delta d/2}\right\vert\left\Vert \frac{1}{2}\{L^\dagger L, \cdot \} \right\Vert_\diamond \notag \\ 
    &\qquad + \left\vert\frac{d^2}{d^2-2}\left[1+\frac{e^{-\Delta d'}-1}{\Delta d'}\right] - \frac{2}{d^2-2}\left[1+\frac{e^{-\Delta d/2}-1}{\Delta d/2}\right]\right\vert \left\Vert L\cdot L^\dagger\right\Vert_\diamond \label{eq:tmp4} \\
    &\leq \left\vert1+\frac{e^{-\Delta d/2}-1}{\Delta d/2}\right\vert\left\Vert L^\dagger L\right\Vert_\infty \notag \\ 
    &\qquad + \left\vert\frac{d^2}{d^2-2}\left[1+\frac{e^{-\Delta d'}-1}{\Delta d'}\right] - \frac{2}{d^2-2}\left[1+\frac{e^{-\Delta d/2}-1}{\Delta d/2}\right]\right\vert \left\Vert L^\dagger L\right\Vert_\infty \label{eq:tmp-pf-1} \\
    &= \left(1+\frac{e^{-\Delta d/2}-1}{\Delta d/2} + \frac{d^2}{d^2-2}\left[1+\frac{e^{-\Delta d'}-1}{\Delta d'}\right] - \frac{2}{d^2-2}\left[1+\frac{e^{-\Delta d/2}-1}{\Delta d/2}\right] \right)\left\Vert L\right\Vert_\infty^2 \label{eq:tmp-pf-2}\\
    &= \frac{d^2}{d^2-2}\left(\left[1+\frac{e^{-\Delta d'}-1}{\Delta d'}\right] - \left[1+\frac{e^{-\Delta d/2}-1}{\Delta d/2}\right]\right)\left\Vert L^\dagger L\right\Vert_\infty \label{eq:tmp5} \\
    &\leq\frac{d^2}{d^2-2}\left(\frac{\Delta d'}{2}-\frac{\Delta d}{4}\right) \left\Vert L^\dagger L\right\Vert_\infty \label{eq:tmp6} \\
    &=\frac{\Delta d}{4}\left\Vert L\right\Vert_\infty^2.
\end{align}
\eqref{eq:tmp4} follows from the triangular inequality, \eqref{eq:tmp-pf-1} follows from derivations in~\cite{go2025SamplebasedHamiltonianLindbladian}, \eqref{eq:tmp-pf-1} follows because $f(x)$ is monotonically increasing, and finally, \eqref{eq:tmp6} follows because $f(x)$ is Lipschitz continuous with constant $\frac{1}{2}$, given that $f'(0)=1/2$ and $f'(x)$ is monotonically decreasing for $x\geq 0$.

\section{Typical-case bound for subgaussian Lindblad operators}
\label{app:subgaussian-average-case}
For completeness, we provide a subgaussian extension of the typical-case bound. Here, $\|\cdot\|_{\psi_2}$ denotes the standard subgaussian norm, i.e.,
\begin{equation}
    \|X\|_{\psi_2} :=
\inf\left\{
t>0 :
\mathbb{E}\!\left[e^{X^2/t^2}\right]\le 2
\right\},
\end{equation} and for a complex random variable $X$, we use $\|X\|_{\psi_2}\coloneqq \||X|\|_{\psi_2}$.
\begin{theorem}[Subgaussian typical-case WML bound]
\label{thm:subgaussian-average-case-WML}
Let $d\ge2$, and let $L'\in\mathbb C^{d\times d}$ be a random matrix whose entries are independent, mean-zero, unit-variance complex subgaussian random variables satisfying $\|L'_{ij}\|_{\psi_2}\le K$, and define $L = L' / \left\|L'\right\|_2$.
Then there exists a constant $C_K>0$, depending only on the subgaussian norm bound $K$, such that for any $t\ge0$, $\varepsilon\in(0,1)$, and $\delta\in(0,1)$, the typical case sample complexity given by Def.~\ref{def: approx sample complexity} satisfies
\begin{equation}
    \overline{n}_d(t,\varepsilon,\delta) \le C_K \frac{t^2}{\varepsilon} \left(1 + \frac{\log(1/\delta)}{d}\right).
\end{equation}

In particular, for fixed $K$ and $\delta$, as the system dimension $d$ increases, the correction term $\log(1/\delta)/d$ vanishes asymptotically, and the sample complexity scales as
\begin{equation}
\overline{n}_d(t,\varepsilon,\delta) = O\!\left(\frac{t^2}{\varepsilon}\right).
\end{equation}
\end{theorem}
The proof of Theorem~\ref{thm:subgaussian-average-case-WML} relies on the following concentration bound for the normalized operator norm of a random Lindblad operator.

\begin{proposition}[Subgaussian concentration of the normalized operator norm]
\label{prop:subgaussian-concentration}
Let $d\ge2$, and let $L' \in \mathbb{C}^{d \times d}$ be a random matrix whose entries are independent, mean-zero, unit-variance complex subgaussian random variables satisfying $\|L'_{ij}\|_{\psi_2}\le K$, and define $L = L' / \left\|L'\right\|_2$. 
Then there exists a constant $C_K>0$, depending only on the subgaussian norm bound $K$, such that for any $\delta\in(0,1)$, with probability at least $1-\delta$,
\begin{equation}
\left\|L\right\|_\infty^2 \le C_K \left( \frac{1}{d} + \frac{\log(1/\delta)}{d^2} \right).
\end{equation}
\end{proposition}

\begin{proof}[Proof of Theorem~\ref{thm:subgaussian-average-case-WML}]
By Theorem~\ref{cor:samp-comp-upper-bnd-WML}, for each fixed Lindblad operator $L$, it suffices to take a number of WML samples bounded by
\begin{equation}
    n_d(t,\varepsilon)
    \le
    \frac{2d+3}{8}
    \left\|L\right\|_\infty^2
    \frac{t^2}{\varepsilon}
    \le
    d\left\|L\right\|_\infty^2
    \frac{t^2}{\varepsilon}.
\end{equation}
From Proposition~\ref{prop:subgaussian-concentration}, we have that with probability at least $1-\delta$,
\begin{equation}
\frac{\left\|L'\right\|_\infty^2}{\left\|L'\right\|_2^2} \le
C_K \left( \frac{1}{d} + \frac{\log(1/\delta)}{d^2} \right).
\end{equation}
Substituting this bound into the preceding WML sample bound gives, with probability at least $1-\delta$,
\begin{equation}
n_d(t,\varepsilon) \le C_K \frac{t^2}{\varepsilon}
\left( 1 + \frac{\log(1/\delta)}{d} \right).
\end{equation}

By Def.~\ref{def: approx sample complexity}, this implies that the typical sample complexity satisfies
\begin{equation}
\overline{n}_d(t,\varepsilon,\delta)
\le C_K \frac{t^2}{\varepsilon} \left( 1 + \frac{\log(1/\delta)}{d} \right).
\end{equation}
As a result, for fixed $K$ and $\delta$, the correction term $\log(1/\delta)/d$ vanishes asymptotically as $d$ grows, yielding
\begin{equation}
\overline{n}_d(t,\varepsilon,\delta) = O\!\left(\frac{t^2}{\varepsilon}\right),
\end{equation}
which completes the proof.
\end{proof}

\medskip

\begin{proof}[Proof of Proposition~\ref{prop:subgaussian-concentration}]
    We estimate the operator norm of the Frobenius-normalized random operator $L=L'/\|L'\|_2$. Since $\left\|L\right\|_\infty^2 = \frac{\|L'\|_\infty^2}{\|L'\|_2^2}$, it suffices to upper bound $\|L'\|_\infty$ and lower bound $\|L'\|_2^2$.
    For the numerator, the standard operator-norm bound for subgaussian random matrices gives universal constants $C,c>0$ such that, for every $u\ge0$~\cite[Theorem~4.4.3]{vershynin2018HighDimensionalProbabilityIntroduction},
    \begin{equation}
        \mathbb{P}\!\left(\|L'\|_\infty > C K \sqrt d + u\right)
        \le 2\exp\!\left(-c\frac{u^2}{K^2}\right).
    \end{equation}
    For the denominator, since the entries have unit variance, $\mathbb{E}\|L'\|_2^2=d^2$. 
    Moreover, the square of a subgaussian random variable is sub-exponential, and a standard Hanson--Wright concentration bound~\cite{hanson1971BoundTailProbabilities} gives a universal constant $c'>0$ such that, for every $\eta\in(0,1)$,
    \begin{equation}
        \mathbb{P}\!\left(\|L'\|_2^2 < (1-\eta)d^2 \right)
        \le2\exp\!\left(-c'\frac{\eta^2d^2}{K^4}\right).
    \end{equation}
    Equivalently, this follows by applying the Hanson--Wright inequality to $\|L'\|_2^2=\braket{L'}{L'}$~\cite[Theorem~6.2.2]{vershynin2018HighDimensionalProbabilityIntroduction}.
    
    Let us define
    \begin{equation}
        X(u,\eta)=\frac{(CK\sqrt d+u)^2}{(1-\eta)d^2}.
    \end{equation}
    By the preceding two tail bounds and the union bound, we have
    \begin{align}
        \mathbb{P}\!\left(\frac{\|L'\|_\infty^2}{\|L'\|_2^2} > X(u,\eta)\right)
        &\leq \mathbb{P}\!\left(\left\|L'\right\|_\infty > CK\sqrt d+u\right)+
        \mathbb{P}\!\left(\left\|L'\right\|_2^2 < (1-\eta)d^2\right)\nonumber\\
        &\le 2\exp\!\left(-c\frac{u^2}{K^2}\right)+ 2\exp\!\left(-c'\frac{\eta^2d^2}{K^4}\right).
    \end{align}
    For a given $\delta \in (0,1)$, choose
    \begin{equation}
        u = C_1 K \sqrt{\log(4/\delta)}, \qquad
        \eta = C_2 K^2 \frac{\sqrt{\log(4/\delta)}}{d},
    \end{equation}
    where $C_1,C_2>0$ are universal constants chosen sufficiently large so that the right-hand side is at most $\delta$.
    Using $(a+b)^2 \le 2a^2 + 2b^2$ and the bound $\frac{1}{1-\eta} \le 2$ for $\eta \le 1/2$, we obtain
    \begin{align}    
        X(u,\eta) &\le \frac{2(CK\sqrt d+u)^2}{d^2}\\
        &\le \frac{4C^2K^2}{d}+\frac{4u^2}{d^2}\\
        &\le C_K \left( \frac{1}{d} + \frac{\log(1/\delta)}{d^2} \right),
    \end{align}
    where $C_K>0$ depends only on $K$.
    Therefore, with probability at least $1-\delta$,
    \begin{equation}
        \left\|L\right\|_\infty^2=\frac{\|L'\|_\infty^2}{\|L'\|_2^2}
        \le C_K\left(\frac{1}{d}+\frac{\log(1/\delta)}{d^2}\right).
    \end{equation}
    This completes the proof.
\end{proof}

\section{Proof of Proposition~\ref{prop:rankone_lower}}
\label{app:worst_prop_proof}
We prove Proposition~\ref{prop:rankone_lower} by evaluating the WML error for a particular rank-one Lindblad operator and input state.
The argument proceeds in four steps.
We define the jump operator as 
\begin{equation}
L\coloneqq |u\rangle\langle u|,
\end{equation}
where $\ket{u}$ is a normalized pure state, i.e., $\langle u | u\rangle =1$. We also take the initial state $\rho = \ket{u} \bra{u}$. The exact evolution of the state under the Lindbladian leads to
\begin{equation}
\rho_{\mathrm{ex}}(t)\coloneq e^{\mathcal Lt}[|u\rangle\langle u|],
\end{equation}
where the WML algorithm leads to the state
\begin{equation}
\rho_{\mathrm{WML}}^{(n)}(t)\coloneq \left(\widetilde{e^{\mathcal Lt/n}}\right)^n[|u\rangle\langle u|].
\end{equation}

\subsection{Exact Lindbladian action}
Since $L= \ket{u}\bra{u} = L^\dagger L$, we have
\begin{equation}
\mathcal L(\ket{u}\bra{u}) = L \ket{u}\bra{u} L^{\dagger} -\frac{1}{2}\{L^{\dagger}L, \ket{u}\bra{u} \}=0.
\end{equation}
Therefore, under the exact Lindbladian dynamics, the initial state remains unchanged, as
\begin{equation}
\rho_{\mathrm{ex}}(t) = \ket{u}\bra{u},
\label{eq:rho_ex_projector}
\end{equation}
for all $t\geq 0$.

\subsection{Exact $n$-step WML state on the invariant two-dimensional space}
From the analytic one-step formula for WML derived in Appendix~\ref{appendix:main-upper-proof}, we have
\begin{equation}
\widetilde{e^{\mathcal L\Delta}}(\rho) = \rho + a\,\frac{1}{2}\{\ket{u}\bra{u},\rho\}
+ b\,L\rho L^\dagger + c\,\mathcal D_1[L\rho L^\dagger],
\label{eq:one_step_wml_exact}
\end{equation}
where
\begin{equation}
\mathcal D_1(X)\coloneqq \frac{\mathbb I}{d}\operatorname{Tr}(X),
\end{equation}
and
\begin{align}
&a = (e^{T\Delta})_{20} = \frac{e^{-d\Delta/2}-1}{d/2}, \\
&(e^{T\Delta})_{10} =
-\frac{e^{-(d-1/d)\Delta}-1}{d-1/d}, \\
&b= \frac{2}{d^2-2}(e^{T\Delta})_{20}
+ \frac{d^2}{d^2-2}(e^{T\Delta})_{10}, \\
&c= -\frac{d^2}{d^2-2}\Bigl((e^{T\Delta})_{20}+(e^{T\Delta})_{10}\Bigr).
\end{align}
Applying \eqref{eq:one_step_wml_exact} to $\ket{u}\bra{u}$ and $\mathbb I$, we obtain
\begin{align}
\widetilde{e^{\mathcal L\Delta}}(\ket{u}\bra{u})
&=(1+a+b)\ket{u}\bra{u}+\frac{c}{d}\mathbb I, \\
\widetilde{e^{\mathcal L\Delta}}(\mathbb I)
&=(a+b)\ket{u}\bra{u}+\left(1+\frac cd\right)\mathbb I.
\end{align}
Therefore, the space $\operatorname{span}\{\ket{u}\bra{u},\mathbb I\}$ is invariant under one-step WML.

Now, for each $n\ge0$, we write 
\begin{equation}
\rho_{\mathrm{WML}}^{(n)}(t)=y_n\ket{u}\bra{u}+z_n\mathbb I
\end{equation}
for real coefficients $y_n,z_n$. The coefficient vector satisfies
\begin{equation}
\begin{pmatrix}
y_{n+1}\\ z_{n+1}
\end{pmatrix}
= M_\Delta
\begin{pmatrix}
y_n\\ z_n
\end{pmatrix},
\qquad
M_\Delta\coloneqq 
\begin{pmatrix}
1+s & s\\[1mm]
\kappa & 1+\kappa
\end{pmatrix},
\label{eq:block_matrix_update}
\end{equation}
where
\begin{equation}
s\coloneqq a+b, \qquad
\kappa\coloneqq \frac{c}{d}.
\end{equation}
The initial vector is
\begin{equation}
\begin{pmatrix}
y_0\\ z_0
\end{pmatrix}
=
\begin{pmatrix}
1\\ 0
\end{pmatrix}.
\end{equation}

Now define
\begin{equation}
\lambda\coloneqq 1+s+\kappa=1+a+b+\frac cd.
\end{equation}
A direct computation shows that $M_\Delta$ has eigenvalues $1$ and $\lambda$, with corresponding eigenvectors
\begin{equation}
v_1=
\begin{pmatrix}
1\\ -1
\end{pmatrix},
\qquad
v_2=
\begin{pmatrix}
s\\ \kappa
\end{pmatrix}.
\end{equation}
Since
\begin{equation}
\begin{pmatrix}
1\\ 0
\end{pmatrix}
= \frac{\kappa}{\lambda-1}
\begin{pmatrix}
1\\ -1
\end{pmatrix}
+ \frac{1}{\lambda-1}
\begin{pmatrix}
s\\ \kappa
\end{pmatrix},
\end{equation}
it follows that
\begin{equation}
M_\Delta^n
\begin{pmatrix}
1\\ 0
\end{pmatrix}
= \frac{\kappa}{\lambda-1}
\begin{pmatrix}
1\\ -1
\end{pmatrix}
+ \frac{\lambda^n}{\lambda-1}
\begin{pmatrix}
s\\ \kappa
\end{pmatrix}.
\end{equation}
Therefore
\begin{equation}
\begin{pmatrix}
y_n\\ z_n
\end{pmatrix}
= \frac{1}{\lambda-1}
\begin{pmatrix}
\kappa+s\lambda^n\\[1mm]
\kappa(\lambda^n-1)
\end{pmatrix}.
\label{eq:yn_zn_exact}
\end{equation}
In particular,
\begin{equation}
z_n = \frac{\kappa(\lambda^n-1)}{\lambda-1}
= \frac{c}{d}\,\frac{\lambda^n-1}{\lambda-1}.
\label{eq:z_rec_final}
\end{equation}

\subsection{Trace-distance lower bound}
Let
\begin{equation}
\Omega_n(t)\coloneqq \rho_{\mathrm{WML}}^{(n)}(t)-\rho_{\mathrm{ex}}(t).
\end{equation}
By \eqref{eq:rho_ex_projector}, the exact state $\rho_{\mathrm{ex}}(t) = \ket{u}\bra{u}$ is stationary.
On the other hand,
\begin{equation}
\rho_{\mathrm{WML}}^{(n)}(t)=y_n\ket{u}\bra{u}+z_n\mathbb I
\end{equation}
has restriction $z_nI_{S^\perp}$ on the orthogonal complement $S^\perp$. Hence
\begin{equation}
\Omega_n(t)\big|_{S^\perp}=z_nI_{S^\perp}.
\end{equation}
Since $\dim(S^\perp)=d-1$, we obtain
\begin{equation}
\frac{1}{2}\|\Omega_n(t)\|_1 \ge \frac{d-1}{2}\,|z_n|.
\label{eq:error_zn}
\end{equation}

We now show that $z_n>0$ for all sufficiently large $n$.
From the Taylor expansions
\begin{equation}
a=-\Delta+\frac d4\Delta^2+O(\Delta^3),
\qquad
b=\Delta-\frac d2\Delta^2+O(\Delta^3),
\qquad
c=\frac d4\Delta^2+O(\Delta^3),
\end{equation}
we obtain
\begin{equation}
\lambda = 1+a+b+\frac cd
= 1-\frac{d-1}{4}\Delta^2+O(\Delta^3)<1
\end{equation}
for a sufficiently large $n$. Hence, $\lambda^n-1<0$ and $\lambda-1<0$, so
\begin{equation}
\frac{\lambda^n-1}{\lambda-1}>0.
\end{equation}
Moreover, the explicit formula for $c$ implies $c>0$ for $\Delta>0$ and $d\ge2$.
Therefore, \eqref{eq:z_rec_final} implies that
\begin{equation}
z_n>0.
\end{equation}
So \eqref{eq:error_zn} becomes
\begin{equation}
\frac{1}{2}\|\Omega_n(t)\|_1
\ge
\frac{d-1}{2}\,z_n.
\label{eq:error_zn_positive}
\end{equation}

\subsection{Asymptotic expansion of $z_n$}
From the coefficient formulas, we obtain
\begin{equation}
\lambda = 1-\frac{d-1}{4}\Delta^2+O(\Delta^3),
\qquad
\kappa=\frac{c}{d}=\frac14\Delta^2+O(\Delta^3).
\end{equation}
Then, by writing $\lambda=1-\frac{d-1}{4}\Delta^2+O(\Delta^3)=1+x$ with $x=O(\Delta^2)$ and using the expansion $\log(1+x)=x-\frac{x^2}{2}+O(x^3)$ gives
\begin{equation}
\log\lambda=-\frac{d-1}{4}\Delta^2+O(\Delta^3).
\end{equation}
Multiplying by $n$ and using $\Delta=t/n$, we obtain
\begin{equation}
n\log\lambda = -\frac{d-1}{4}\frac{t^2}{n}+O_{d,t}(n^{-2}),
\end{equation}
which leads to
\begin{equation}
\lambda^n = \exp{-\frac{d-1}{4}\frac{t^2}{n} + O_{d,t}(n^{-2})}=1-\frac{d-1}{4}\frac{t^2}{n} + O_{d,t}(n^{-2}).
\end{equation}
We also note that 
\begin{equation}
\frac{\lambda^n-1}{\lambda-1} = n + O_{d,t}(1).
\end{equation}
Substituting these into \eqref{eq:z_rec_final}, we obtain
\begin{equation}
z_n = \frac{t^2}{4n} + O_{d,t}(n^{-2}).
\label{eq:zn_asymptotic}
\end{equation}
Combining \eqref{eq:error_zn_positive} and \eqref{eq:zn_asymptotic}, we find
\begin{equation}
\frac{1}{2}\|\rho_{\mathrm{WML}}^{(n)}(t)-\rho_{\mathrm{ex}}(t)\|_1
\ge \frac{d-1}{8}\,\frac{t^2}{n} + O_{d,t}(n^{-2}).
\end{equation}
Therefore, there exists $N_{d,t}\in\mathbb N$ such that for all $n\ge N_{d,t}$,
\begin{equation}
\frac{1}{2}\|\rho_{\mathrm{WML}}^{(n)}(t)-\rho_{\mathrm{ex}}(t)\|_1
\ge \frac{d-1}{16}\,\frac{t^2}{n}.
\label{eq:error_lower_rzero}
\end{equation}

Now, we define the following sample complexities
\begin{equation}
n_\varepsilon = n_{\mathrm{WML}}(t,\varepsilon;L,\rho)
\coloneqq  
\inf\left\{ n\in\mathbb N:\frac12\left\|\left({\widetilde{e^{\mathcal{L}[L]t}}} - e^{\mathcal{L}[L]t}\right)[\rho]\right\|_1 \le \varepsilon \right\}.
\end{equation}
and
\begin{equation}
n_{\mathrm{WML}}(t,\varepsilon)
\coloneqq  
\inf\left\{ n\in\mathbb N:\sup_{\mathcal{L}\in\mathcal{B}(\mathcal{H}_d)}\frac12\left\|\left({\widetilde{e^{\mathcal{L}t}}} - e^{\mathcal{L}t}\right)\right\|_\diamond \le \varepsilon \right\}.
\end{equation}
From these definitions, we note that 
\begin{equation}
\frac{1}{2}\|\rho_{\mathrm{WML}}^{(n_\varepsilon)}(t)-\rho_{\mathrm{ex}}(t)\|_1
\le \varepsilon,
\end{equation}
and
\begin{equation}
n_{\mathrm{WML}}(t,\varepsilon;L, \rho)\leq n_{\mathrm{WML}}(t,\varepsilon)
\end{equation}
For sufficiently small $\varepsilon>0$, we have $n_\varepsilon\ge N_{d,t}$, and hence \eqref{eq:error_lower_rzero} implies
\begin{equation}
\frac12\|\rho_{\mathrm{WML}}^{(n_\varepsilon)}(t)-\rho_{\mathrm{ex}}(t)\|_1
\ge \frac{d-1}{16}\,\frac{t^2}{n_\varepsilon}.
\end{equation}
Combining the two inequalities, we obtain
\begin{equation}
\varepsilon \ge \frac{d-1}{16}\,\frac{t^2}{n_\varepsilon}.
\end{equation}
Consequently, the sample complexity is lower bounded as
\begin{equation}
n_{\mathrm{WML}}(t,\varepsilon)\geq n_\varepsilon \ge \frac{d-1}{16}\,\frac{t^2}{\varepsilon}
\end{equation}
for sufficiently small $\varepsilon>0$.
Finally, for $d\ge2$, we have $d-1\ge d/2$. This leads to
\begin{equation}
n_{\mathrm{WML}}(t,\varepsilon)
\ge
\frac{d}{32}\,\frac{t^2}{\varepsilon}
\qquad \text{for sufficiently small }\varepsilon>0,
\end{equation}
which completes the proof.

\section{Alternative realization of WML via Stinespring isometry}

Physical realization of WML Lindbladian evolution $\mathcal{M}$ may not be straightforward in quantum computers, where operations are inherently unitary. Recently, it was shown how to simulate $e^{\mathcal{M}\Delta}$ \cite[Appendices B and C]{sims2025DigitalQuantumSimulations}. In this section, we propose another method that utilizes the structure of $M$.

One example that bypasses the problem is to consider the Stinespring isometry of the dynamics. Since the WML Lindbladian only contains a single jump operator $M$, a single ancillary qubit is enough to realize the Stinespring isometry. Formally, the isometry is defined as $V_{\delta, M}\coloneqq U_{\delta, M} (\mathbb{I}_{0123}\otimes\ket{0}_a)$ such that $e^{\Delta\mathcal{M}}=\operatorname{Tr}_{a}V_{\delta, M}(~\cdot~)V_{\delta, M}^\dagger$. Here, $U_{\delta, M}$ is a strong evolution of a block-encoded Hamiltonian, sometimes referred to as \emph{J-matrix}~\cite{cleve2019EfficientQuantumAlgorithms, childs2023EfficientSimulationSparse, sims2025DigitalQuantumSimulations}:
\begin{equation}
    U_{\delta, M} = \exp\!\left[{-i\sqrt{\delta}\begin{bmatrix}
        0 & M^\dagger \\ M & 0
    \end{bmatrix}}\right].
\end{equation} 

Again, due to the repetitive structure of higher orders of $M$, we can exactly express the unitary with a few terms:
\begin{equation}
    U_{\delta, M} = \begin{bmatrix}
        \mathbb{I}-\left(\frac{1-\cos\sqrt{\delta d}}{d}\right)M^\dagger M & ~ * ~\\ -i\frac{\sin\sqrt{\delta d}}{\sqrt{d}}M & ~ * ~
    \end{bmatrix}.
\end{equation}
Therefore, letting $\Delta=2\tfrac{1-\cos\sqrt{\delta d}}{d}$, the Stinespring isometry for a single step evolves the input state according to the WML Lindbladian operator up to first order:
\begin{align}
    \operatorname{Tr}_{a}V_{\delta, M}\rho V_{\delta, M}^\dagger & = \rho+\Delta\mathcal{M}(\rho)+\frac{\Delta^2d}{4}\left(\frac{M^\dagger M\rho M^\dagger M}{d}-M\rho M^\dagger\right),\\
    \operatorname{Tr}_{23a}V_{\delta, M}\rho \otimes \pi_L V_{\delta, M}^\dagger & = \rho+\Delta\mathcal{L}(\rho)+\frac{\Delta^2d}{4}(\mathcal{D}_1-\mathcal{I})(L\rho L^\dagger).
\end{align}
We refer to Appendix~\ref{appendix:main-upper-proof} for the detailed derivation. Again, using the cyclic property of trace, we conclude for the Stinespring expansion, 
\begin{equation}
    \widetilde{e^{\mathcal{L}\Delta}}=\mathcal{I} + \Delta(\mathcal{I}-\mathcal{D}_1)\widetilde{\mathcal{L}}
\end{equation}
for  $\widetilde{\mathcal{L}}[\rho]=\mathcal{L}[\rho]-\frac{\Delta d}{4}L\rho L^\dagger$, which satisfies
\begin{equation}
    \left\Vert\widetilde{\mathcal{L}}-\mathcal{L}\right\Vert_\diamond\leq\frac{\Delta d}{4}\left\Vert L\right\Vert_\infty^2.
\end{equation}
This again leads to Lemma~\ref{lem:main-upper-lemma}, such that the sample complexity upper bound remains $O(d\left\|L\right\|_\infty^2t^2/\varepsilon)$. The scaling behavior remains the same even when using a more advanced method for simulating $e^{\Delta\mathcal{M}}$, e.g., polynomial-depth approximation via direct decomposition $M$, or LCU~\cite[Appendices B and C]{sims2025DigitalQuantumSimulations}.

\end{document}